\documentclass[journal]{IEEEtran}
\usepackage{amssymb}
\usepackage{setspace}
\usepackage{amsthm}
\usepackage{amsmath}
\usepackage{color}

\newtheorem{thm}{Theorem}
\newtheorem{prop}{Proposition}
\newtheorem{coro}{Corollary}

\newtheorem{remark}{Remark}
\newcommand{\D}[1]{\textrm{d} {#1}}
\newcommand{\p}[1]{\textrm{Pr}\left({#1}\right)}
\newcommand{\pp}[1]{\textrm{Pr}({#1})}
\newcommand{\ee}[1]{\textrm{E}\left({#1}\right)}
\newcommand{\eee}[1]{\textrm{E}({#1})}
\newcommand{\vv}[1]{\textrm{Var}\left({#1}\right)}
\newcommand{\vvv}[1]{\textrm{Var}({#1})}
\newcommand{\I}[1]{\textbf{1}_{\{{#1}\}}}
\def\eqd{\,{\buildrel D \over =}\,}

\ifCLASSINFOpdf
  \usepackage[pdftex]{graphicx}
\else
  \usepackage[dvips]{graphicx}
\fi

\begin{document}
\title{On Asymptotic Statistics for Geometric Routing Schemes in Wireless Ad-Hoc Networks}

\author{Armin Banaei, Daren B.H. Cline, Costas N. Georghiades, and Shuguang Cui% <-this % stops a space
\thanks{A. Banaei, C. Georghiades, and S. Cui are with the Department
of Electrical and Computer Engineering, Texas A\&M University.}% <-this % stops a space
\thanks{D. Cline is with the Department of Statistics, Texas A\&M University.}}

\maketitle

\begin{abstract}

In this paper we present a methodology employing statistical analysis and stochastic geometry to study geometric routing schemes in wireless ad-hoc networks. In particular, we analyze the network layer performance of one such scheme, the random $\frac{1}{2}$disk routing scheme, which is a localized geometric routing scheme in which each node chooses the next relay randomly among the nodes within its transmission range and in the general direction of the destination. The techniques developed in this paper enable us to establish the asymptotic connectivity and the convergence results for the mean and variance of the routing path lengths generated by geometric routing schemes in random wireless networks. In particular, we approximate the progress of the routing path towards the destination by a Markov process and determine the sufficient conditions that ensure the asymptotic connectivity for both dense and large-scale ad-hoc networks deploying the random $\frac{1}{2}$disk routing scheme. Furthermore, using this Markov characterization, we show that the expected length (hop-count) of the path generated by the random $\frac{1}{2}$disk routing scheme normalized by the length of the path generated by the ideal direct-line routing, converges to $3\pi/4$ asymptotically. Moreover, we show that the variance-to-mean ratio of the routing path length converges to $9\pi^2/64-1$ asymptotically. Through simulation, we show that the aforementioned asymptotic statistics are in fact quite accurate even for finite granularity and size of the network.
\end{abstract}

\begin{IEEEkeywords}
Geometric Routing Schemes, Asymptotic Network Connectivity, Asymptotic Path Length Statistics, Statistical Analysis, Stochastic Geometry, Markov Process.
\end{IEEEkeywords}

% \fi
\IEEEpeerreviewmaketitle

\section{Introduction}
\label{sec:Introduction}

\IEEEPARstart{A}{} wireless ad-hoc network consists of autonomous wireless nodes that collaborate on communicating information in the absence of a fixed infrastructure. Each of the nodes might act as a source/destination node or as a relay. Communication occurs between a source-destination pair through a single-hop transmission if they are close enough, or through multi-hop transmissions over intermediate relaying nodes if they are far apart. The selection of relaying nodes along the multi-hop path is governed by the adopted routing scheme.

The conventional method to establish a routing path between a given source-destination pair is through exchanges of control packets containing the complete network topology information \cite{Dijkstra}, which creates scalability issues when the network size becomes large. One way to reduce the overhead for global topology inquiries is to build routes on demand via flooding techniques \cite{AODV}. However, such routing protocols essentially suffer from a similar issue of large signaling overheads. To deal with the above issues, Takagi and Kleinrock \cite{MFR} introduced the first geographical (or position-based) routing scheme, coined as Most Forward within Radius (MFR), based on the notion of progress:\footnote{It should be noted that the reduction in complexity comes at the cost of knowing the location of the neighboring nodes in addition to that of the destination.} Given a transmitting node $S$ and a destination node $Dst$, the progress at relay node $V$ is defined as the projection of the line segment $SV$ onto the line connecting $S$ and $Dst$. In MFR, each node forwards the packet to the neighbor with the largest progress (e.g., node $V_2$ in Fig. \ref{fig:GFRs}), or discards the packet if none of its neighbors are closer to the destination than itself. There are some other variants of the geographical routing scheme in the literature \cite{NFP}--\cite{Zorzi03}, which are similar to MFR. In \cite{NFP}, the authors introduced the Nearest Forward Progress (NFP) method that selects the nearest neighbor of the transmitter with forward (positive) progress (e.g., node $V_1$ in Fig. \ref{fig:GFRs}); in \cite{DIR}, the Compass Routing (also referred to as the DIR method) was proposed, where the neighbor closest to the line connecting the sender and the destination is chosen (e.g., node $V_3$ in Fig. \ref{fig:GFRs}); in \cite{Zorzi03}, the authors considered the Shortest Remaining Distance (SRD) method, where the neighbor closet to the destination is selected as the relay (e.g., node $V_4$ in Fig. \ref{fig:GFRs}).

\begin{figure}
  \centering
  \includegraphics[width=2.5 in]{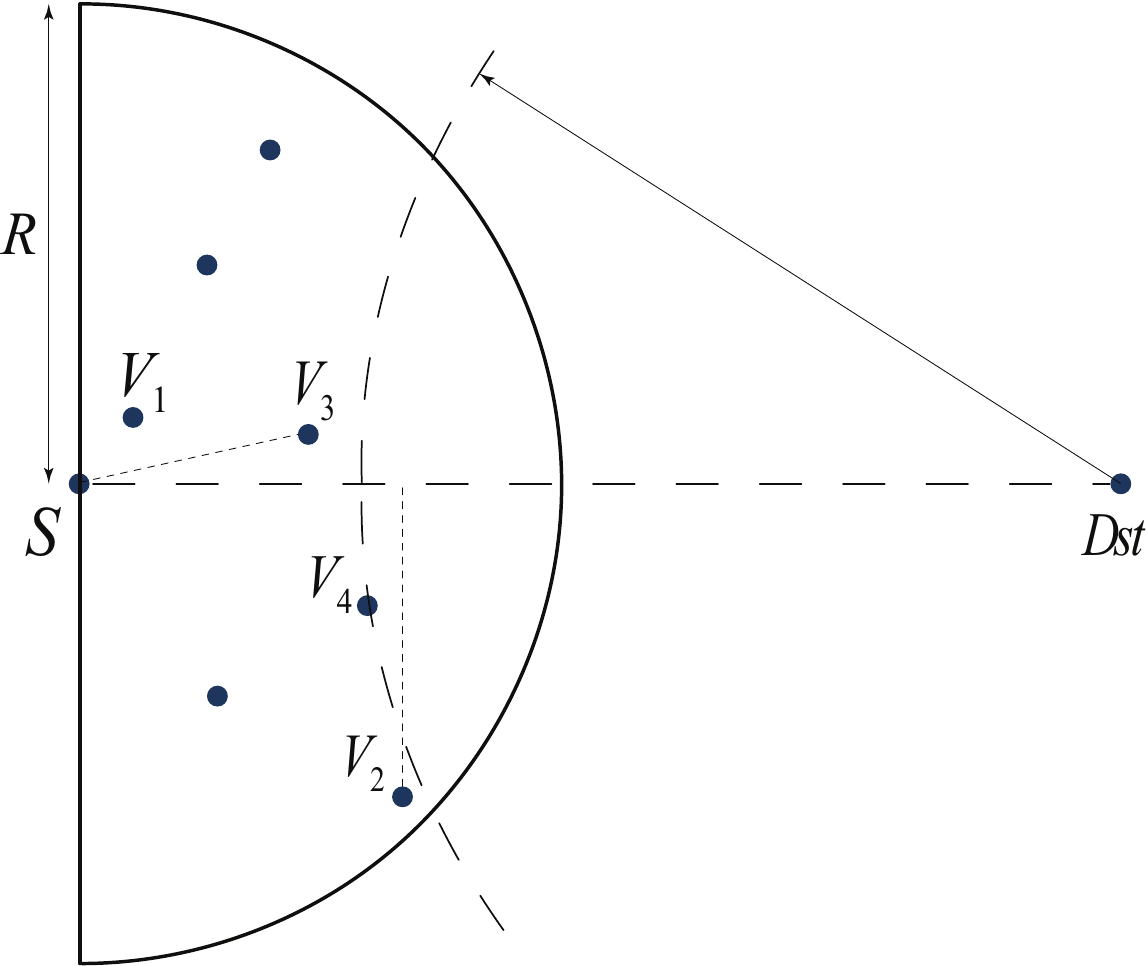}
  \caption{Some variants of geometric routing schemes: The source node $S$ has different choices to find a relay node for further forwarding a message to the destination $Dst$. $V_1$ = Nearest Forward Progress (NFP), $V_2$ = Most Forward within Radius (MFR), $V_3$ = Compass Routing (DIR), $V_4$ = Shortest Remaining Distance (SRD).}
  \label{fig:GFRs}
\end{figure}

Geographical routing protocols might fail for some network configurations due to dead-ends or routing loops. In these cases, alternative routing strategies, such as route discovery based on flooding \cite{flood} and face routing \cite{GPSR} can be deployed. However, it has been shown in \cite{GuptaKumar} that for dense wireless networks, the MFR-like routing strategies will succeed with high probability and there is no need to resort to recovery methods such as face routing. In this paper we study the network layer performance of geographical routing schemes in such dense or large wireless networks; and we expect to observe a similar high-probability successful routing performance (the proof of this claim is presented in Section \ref{subsec:Expected Length}).

Below we present a methodology employing statistical analysis and stochastic geometry to study geometric routing schemes in wireless ad-hoc networks. We consider a wireless ad-hoc network consisting of wireless nodes that are distributed according to a Poisson point process over a circular area, where nodes are randomly grouped in source-destination pairs and can establish direct communication links with other nodes that are within a certain range. We determine the conditions under which, in such a network, all source-destination node pairs are connected via the adopted geographical routing scheme with high probability and quantify the asymptotic statistics (mean and variance) for the length of the generated routing paths. In particular, we focus on a variant of the geographical routing schemes, namely the random $\frac{1}{2}$disk routing scheme, as an example, where each node chooses the next relay uniformly at random among the nodes in its transmission range over a $\frac{1}{2}$disk with radius $R$ oriented towards the destination. This scheme is similar to the geometric routing scheme discussed in \cite{MFR}, in which one of the nodes with forward progress is chosen as a relay at random, arguing that there is a trade-off between progress and transmission success.

We chose the random $\frac{1}{2}$disk routing scheme mainly for tractability and simplicity in mathematical characterization. However, the solution techniques developed in this paper can be used (with some modifications) to study other variants of geographical routing schemes, such as MFR, NFP, DIR, etc, which will be further discussed in Section \ref{sec:generalization}. Moreover, the random $\frac{1}{2}$disk routing scheme can be used to model situations where nodes have partial or imprecise routing information and the locally optimal selection criterion of greedy forwarding schemes fails \cite{shakkotai}, e.g., when nodes have perfect knowledge about their destination locations but imprecise information about their own locations, or when nodes only know the half-plane over which the final destination lies such that randomly forwarding the packet to a node in the general direction of the destination is a plausible choice.

There has been a considerable interest regarding the network connectivity and the average length of the route generated by geographical routing schemes under different network settings \cite{shakkotai}, \cite{GuptaKumar2}--\cite{Baccelli}. The authors in \cite{GuptaKumar2} considered a wireless network that consists of $n$ nodes uniformly distributed over a disc of unit area with each node transmission covering an area of $r(n) = (\log n+c(n))/n$. They show that this network is connected asymptotically with probability one \emph{if and only if} $c(n)\to\infty$ as $n\to\infty$. Although the asymptotic expression that they derived for the sufficient transmission range is similar to ours, their notion of connectivity is quite different from ours. In \cite{GuptaKumar2}, the network is connected as long as it is percolated, i.e., the network contains an infinite-order component, where no constraints are considered for the paths connecting source-destination pairs. However, the routing paths that we consider in this work have more structure such that we need a different proof technique to prove the asymptotic connectivity of the network. Xing et al. showed in \cite{Xing} that the route establishment can be guaranteed between any source-destination pair using greedy forwarding schemes if the transmission radius is larger than twice the sensing radius in a fully covered homogeneous wireless sensor network. In \cite{analys} the authors derived the critical transmission radius to be $\sqrt{\frac{\beta_0 \log n}{n}}$ which ensures network connectivity asymptotically almost surely (a.a.s.) based on the SRD routing method, where $\beta_0 = 1/(2\pi/3-\sqrt{3}/2)$.

In \cite{Bordenave}, Bordenave considered the maximal progress navigation for small world networks and showed that small world navigation is regenerative.\footnote{This routing scheme, unlike ours, assumes nonnegative progress in each hop.} It is shown furthermore in \cite{Bordenave} that as the cardinality of the navigation (or routing) path grows, the expected number of hops converges, without providing an explicit value for the limit. Baccelli et al. \cite{Baccelli} introduced a time-space opportunistic routing scheme for wireless ad-hoc networks which utilizes a self-selection procedure at the receivers. They show through simulations that such opportunistic schemes can significantly outperform traditional routing schemes when properly optimized. Furthermore, they analytically proved the asymptotic convergence of such schemes. In \cite{shakkotai}, Subramanian and Shakkottai studied the routing delay (measured by the expected length of the routing path) of geographic routing schemes when the information available to network nodes is limited or imprecise. They showed that one can still achieve the same delay scaling even with limited information. Note that the asymptotic delay expression derived in \cite{shakkotai} is similar to the one we derive in this paper; however, our proof technique is more constructive and enables us to derive tight bounds for the mean and the variance of the routing-path lengths in a network of arbitrary size, together with the exact expressions for their asymptotes. Moreover, in \cite{shakkotai} the authors presumes that the progress (as defined in \cite{MFR} and described earlier) at nodes along the routing path form a sequence of i.i.d. random variables. However, as we show later (cf. Proposition \ref{prop1}), this assumption may not hold for Poisson distributed networks of arbitrary finite sizes as the distribution of nodes contained in the transmission range of a given node along a routing path depends on the history of the routing path up to this node, i.e., the progress at each hop is history dependent. Hence, it is neither independent nor identically distributed; but we show that, as the size of the network (either density or area) goes to infinity, the conditional distribution of the progresses along the routing path given the two previous hops, in fact, depends asymptotically only on the last hop.

The remainder of this paper is organized as follows. In Section \ref{sec:SystemModel} we introduce the system model and describe the random $\frac{1}{2}$disk routing scheme. Then we define the notion of connectivity based on generic geometric routing schemes and state the main results of the paper in a theorem regarding the connectivity and the statistical performance of the random $\frac{1}{2}$disk routing scheme. In Sections \ref{sec:connectivity} and \ref{sec:Path Length Statistics} we prove the claims made in this theorem. In Section \ref{sec:connectivity}, we establish sufficient conditions on the transmission range that ensure the existence of a relaying node in every direction of a transmitting node for both dense and large-scale networks. In Section \ref{sec:Path Length Statistics}, we study the stochastic properties of the paths generated by the random $\frac{1}{2}$disk routing scheme. Specifically, in Section \ref{subsec:Markov Approximation}, we prove that the routing path progress conditioned on the previous two hops can be approximated with a Markov process. In Section \ref{subsec:Expected Length}, using the Markovian approximation, we derive the asymptotic expression for the expected length, and in Section \ref{subsec:Length Variance} we derive the asymptotic expression for the variance of the length of the random $\frac{1}{2}$disk routing paths. In Section \ref{sec:Simulation Results}, we present some simulation results to validate our analytical results. In Section \ref{sec:generalization}, we present some guidelines on how to generalize the results derived for the random $\frac{1}{2}$disk routing scheme to other variants of the geometric routing schemes. We conclude the paper in Section \ref{sec:conclusion}.

\section{System Model}
\label{sec:SystemModel}

Consider a circular area $A$ over which a network of wireless nodes resides.\footnote{The results will carry over, with some minor considerations, to any convex region with bounded curvature.} Nodes are distributed according to a homogeneous Poisson point process with density $\lambda$. In this work we adopt a continuum model for the network where each node is a zero-dimensional point in a unit-area disk.\footnote{This is due to the asymptotic nature of the results presented in this work. Furthermore, a Poisson point process model for the node locations can be considered on a discrete space of countably infinite isolated points (for instance, lattices). Adapting such a model does not change the nature of the results presented.} As such, network nodes can be located at any geometric locations $(x,y) \in \mathbb{R}^2$ such that $x^2+y^2 \leq \frac{|A|}{\pi}$, where $|A|$ denotes the area of region $A$.

Each node picks a destination node uniformly at random among all other nodes in the network, and operates with a fixed transmission power that can cover a disk of radius $R = R(\lambda,|A|)$.\footnote{As mentioned earlier, we are only interested in the network layer performance of the network; as such, we do not consider physical layer related issues such as interference. However, as a rule of thumb (cf. \cite{GuptaKumar}), to minimize the interference among wireless nodes we are interested in the smallest transmission radius that ensures network connectivity in this paper.}

For a generic geometric routing scheme, when the targeted destination node is out of the one-hop transmission range $R$ of a given transmitting node, the next relay is selected (based on some rules) among the nodes contained in the \emph{relay selection region} (RSR) of the transmitting node, where the RSR, in general, can be any subset of a full disk of radius $R$ centered at the transmitting node. For example, the RSR for all the geometric routing schemes cited in the introduction section is a \emph{$\frac{1}{2}$disk} of radius $R$ centered at the transmitting node and oriented towards the destination (denoted by $\frac{1}{2}$RSR). We define the rule that governs the selection of the next relay in each node's RSR as the \emph{relay selection rule} (RSL). For example, the RSL for MFR is to choose the node with the largest ``progress'' towards the destination among the nodes contained in its $\frac{1}{2}$RSR. We define the progress $x'_V$ at a relay node $V$ as in \cite{MFR}, and described in the introduction section.

We define the network to be \emph{connected} if for any source-destination node pair in the network, there exists a path constructed by a \emph{finite sequence of relay nodes} complying with the RSL, with \emph{high probability};\footnote{According to this definition, the network is connected if starting from any source and choosing relays based on the routing scheme, the destination is reachable with high probability.} henceforth, we call such a relay sequence a \emph{routing path}. Note that a node can potentially act as a \emph{relay} only if it is contained in the RSR of the current transmitting node. For the sake of definition, we claim that the network is connected if the set of network nodes is empty.

In this paper we study a special case of localized geometric routing schemes, namely the \emph{random $\frac{1}{2}$disk routing scheme}, where for each transmitting node $S$ in the network, the next relay $V$ is selected \emph{uniformly at random} among the nodes contained in the $\frac{1}{2}$RSR of $S$. We denote the relay selection rule of the random $\frac{1}{2}$disk routing scheme by rRSL. Observe that according to our routing scheme, the next chosen relay might be farther away from the destination than the current transmitting node.

In the following, we present a theorem that summarizes the main results of this paper on the random $\frac{1}{2}$disk routing scheme, regarding i) the sufficient conditions on $R(\lambda,|A|)$, which ensure the existence of a relaying node in any direction of a particular transmitting node based on a generalized version of $\frac{1}{2}$RSR; ii) the mathematical model describing the routing path; iii) the mean asymptotes of the path-lengths established by the random $\frac{1}{2}$disk routing scheme; iv) the corresponding variance asymptotes; and v) the asymptotic network connectivity with the random $\frac{1}{2}$disk routing scheme. For the generalized version of the $\frac{1}{2}$RSR, we assume that the RSR of a node is a wedge of angle $2\pi \eta$ with radius $R$, where $0 < \eta \leq 1$ (hereafter called $\eta$disk or $\eta$RSR, interchangeably). Hence, the $\frac{1}{2}$RSR is a special case of the $\eta$RSR with $\eta = 1/2$.

%\begin{figure}
% \centering
%  \includegraphics[width=2.8 in]{routingscheme.eps}
%  \caption{The random $\frac{1}{2}$disk routing scheme.}
%  \label{fig:routingscheme}
%\end{figure}

Note that in this paper we define the \emph{length} of a routing path as the number of hops traversed over the routing path between a source and its destination. For notational convenience, we let $N := \lambda |A|$ designate the expected number of nodes in the network region of area $|A|$ and $d = d(N) := \frac{\pi R^2}{|A|}$ denote the normalized area of a full disk with radius $R$ relative to the area of the whole region, such that $dN$ is the expected number of nodes in such a disk. The \emph{asymptotic} nature of the results presented in this paper is due to $N\to\infty$, which can represent results for either large-scale networks (i.e., when $|A|\to\infty$ with a fixed $\lambda$) or dense networks (i.e., when $\lambda\to\infty$ with a fixed $|A|$).

Also, $f(n) = O\left(g(n)\right)$ means that there exist positive constants $c$ and $M$ such that $f(n)/g(n) \leq c$ whenever $n \geq M$, $f(n) = o\left(g(n)\right)$ means that $\lim f(n)/g(n) \to 0$ as $n\to\infty$, $f(n) \sim g(n)$ means that $\lim f(n)/g(n) \to 1$ as $n\to\infty$, and $f(n) = \Theta\left(g(n)\right)$ means that both $f(n) = O\left(g(n)\right)$ and $g(n) = O\left(f(n)\right)$.

\begin{thm}
\label{thm1}
Consider a Poisson distributed wireless network with an average node population $N$ deployed over a circular area $A$. Each node picks
a destination node uniformly at random among all other nodes in the network. Assume all nodes have the same transmission range $R(N)$ that covers a normalized area $d = d(N)$ and let $x'$ be the progress at each node. Choosing $R(N)$ such that $\eta dN + \log d \to +\infty$ as $N\to\infty$, we have
\begin{enumerate}
  \item[i)] the $\eta$disk of each node in the network pointing at any direction in which its targeted destinations may lie contains at least one relaying node asymptotically almost surely (a.a.s.);
  \item[ii)] the routing path progress can be approximated to a ``second-order'' with a Markov process; more specifically, the conditional distribution of the next hop given the previous two hops, asymptotically depends only on the last hop.
  \item[iii)] Using the Markovian approximation, we have that the length $\nu$ of the random $\frac{1}{2}$disk routing path is asymptotically finite with the asymptotic expected value $\ee{\nu} \sim \frac{32}{15}\frac{1}{\sqrt{d}}$; specifically, the expected length of the random $\frac{1}{2}$disk routing path connecting a source-destination pair that is $h$-distance apart satisfies $\ee{\nu\mid h} \sim \frac{h}{\ee{x'}} = \frac{3\pi}{4}\frac{h}{R}$ as $N\to\infty$;
  \item[iv)] the variance-to-mean ratio of the routing path length satisfies $\frac{\vv{\nu}}{\ee{\nu}} \sim \frac{\vv{x'}}{\ee{x'}^2} = \frac{9\pi^2}{61}-1$ as $N\to\infty$;
  \item[v)] the network is asymptotically connected with the random $\frac{1}{2}$disk routing scheme with high probability,
\end{enumerate}
where the expectation is taken over all realizations of the network nodes, source-destination pair assignments, and the routing paths between source-destination pairs.
\end{thm}
\begin{proof}%[Proof Sketch]
Here we only sketch the outline of the proof and present the respective details in the following sections. In Section \ref{sec:connectivity}, we show that for random networks, choosing $R(N)$ such that $\eta dN + \log d \to +\infty$ as $N\to\infty$ guarantees the existence of at least one relaying node in the $\eta$disk of each network node pointing at any directions in which their targeted destinations may lie a.a.s..\footnote{A specific node might act as a relay for multiple source-destination pairs.} To this end, we first derive an upper bound on the probability $\sigma(N)$ that the $\eta$disk of some nodes in the network pointing at some directions is empty. Then we show that choosing $d(N)$ as mentioned before ensures the asymptotic convergence of $\sigma(N)$ to zero as $N \to\infty$. This ensures the existence of a relaying node in every direction of a particular transmitting node and ascertains the possibility of packet delivery to a particular destination from any direction a.a.s..

In Section \ref{sec:Path Length Statistics}, assuming $R(N)$ satisfies the above condition and $N$ is large enough such that there exists a relaying node in every direction of a particular transmitting node with high probability, we prove that the routing path progress conditioned on the previous two hops can be approximated with a Markov process. Using the Markovian approximation, we then derive the asymptotic expressions for the mean and variance of the routing path length generated by the random $\frac{1}{2}$disk routing scheme between a source-destination pair that is $h$-distance apart and show that they are asymptotic to $\frac{h}{\ee{x'}} = \frac{3\pi}{4}\frac{h}{R}$ and $\frac{\vv{x'}}{\ee{x'}^2} \ee{\nu}= \left(\frac{9\pi^2}{61}-1\right)\ee{\nu}$, respectively. Furthermore, we show that the length of the random $\frac{1}{2}$disk routing path connecting a source to its destination is finite asymptotically. This shows that starting from a source and following the random $\frac{1}{2}$disk routing scheme we can reach the destination in finitely many hops with high probability; hence the network is asymptotically \emph{connected} with the random $\frac{1}{2}$disk routing scheme.
\end{proof}

\section{Theorem \ref{thm1}.$i$ Proof: Uniform Relaying Capability}
\label{sec:connectivity}

In this section we derive the sufficient conditions on $R(N)$ that ensure, for any node in the network, its $\eta$disks pointing in any directions over which its targeted destinations may lie contain at least one potential relaying node. To this end, we first characterize the upper bound on the probability $\sigma(N)$ that, for some network nodes, there are certain directions at which their $\eta$disks are empty; we then choose $R$ such that this bound is vanishingly small. In this process, we can distinguish between two types of network nodes based on their distances to the edge of the network: Nodes that are farther than $R$ away from the edge of the network, which we call \emph{interior nodes}, and nodes that are closer than $R$ to the edge of the network, which we call \emph{edge nodes}. For the sake of definition, we assume $\sigma(N) = 0$ when $N=0$.

For interior nodes, it is clear that the node distribution in their $\eta$disks, pointing at any direction, is the same. Therefore, the existence probability of an empty $\eta$disk for an interior node is independent of its targeted destination direction. However, due to the proximity of edge nodes to the boundary of the network, the existence probability of an empty $\eta$disk for an edge node highly depends on its destination orientation. For example, the $\eta$disks that fall partly outside the network region are more likely to be empty than the ones that are fully contained in the network region. Hence, we derive the probabilities of a node having an empty $\eta$disk in some direction separately for the interior nodes and the edge nodes, denoted by $\sigma'(N)$ and $\sigma''(N)$, respectively.

Recall that a $\eta$disk is a wedge of angle $2\pi \eta$ and radius $R$, with $0 < \eta \leq 1$. Each $\eta$disk has an expected number of nodes $\eta dN$. As shown in Section \ref{subsec:calculationofsigma}, the existence probability of an empty $\eta$disk increases as $\eta$ decreases. However, we can show that the expected length of the routing path connecting a source to its destination will decrease as $\eta$ decreases. Hence, there exist a tradeoff between the existence probability of an empty $\eta$disk (i.e., a disconnected node) and the expected length of the routing path between a source-destination pair parameterized by $\eta$. We leave the study of this trade-off to future work and only derive (in Section \ref{sec:Path Length Statistics}) the mean and variance of the path length connecting a source-destination pair when $\eta = 1/2$.

\subsection{Calculation of $\sigma'(N)$}
\label{subsec:calculationofsigma1}

Consider an interior node $x$, fixed for now. Given $i \geq 1$ nodes are in the transmission range of $x$, their directions in reference to $x$ are independent and uniformly distributed on $[0,2\pi]$. The probability that $x$ has an empty $\eta$disk in some direction equals the probability $U_i(\eta)$ that the angle of the widest wedge containing none of these $i$ nodes is at least $2\eta\pi$. It is not difficult to give a simple upper bound on $U_i(\eta)$: Of the $i$ nodes, without loss of generality (W.L.O.G.), we can assume that (at least) one is at one edge of an empty wedge with angle of $2\eta \pi$, while the other $i-1$ are distributed independently and uniformly in the remainder of the full transmission disk, as shown in Fig. \ref{fig:emptyarc}. Hence, we obtain $U_i(\eta) \leq i(1-\eta)^{i-1}$, for $i\geq 1$. Of course, if $i=0$ the probability is $U_0(\eta) = 1$.
\begin{figure}[ht]
  \centering
  \includegraphics[width=2 in]{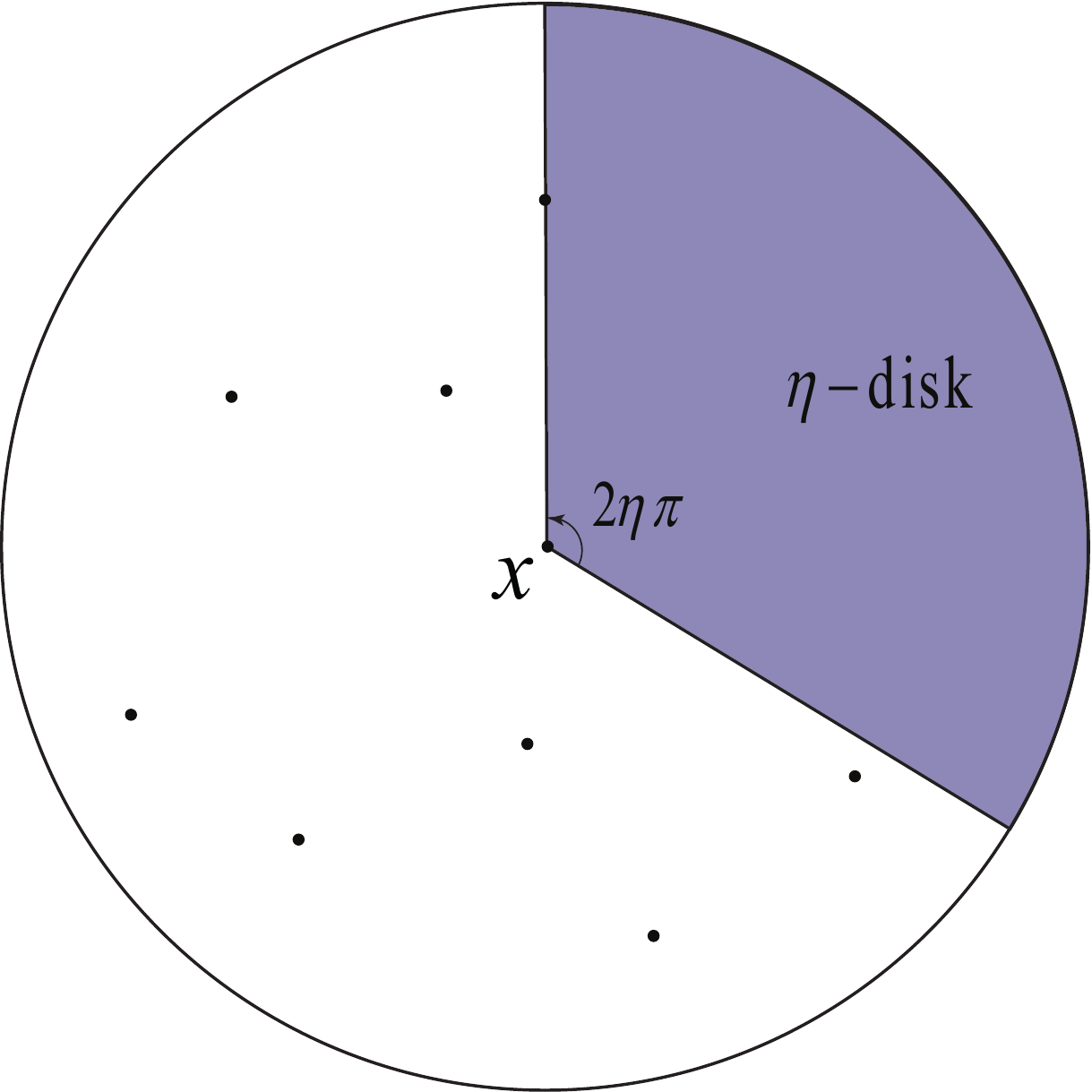}
  \caption{A realization for which the widest wedge between the nodes is of an angle at least $2\eta \pi$.}
  \label{fig:emptyarc}
\end{figure}

One can obtain a more precise expression for $U_i(\eta)$  using results in \cite{Mardia}, page 188:
\begin{equation*}
U_i(\eta) = \sum_{k = 1}^{\min\{\lfloor 1/\eta\rfloor, i\}} (-1)^{k-1}  \binom{i}{k} (1-k\eta)^{i-1} \leq i(1-\eta)^{i-1}\,,
\end{equation*}
for $i\geq 1$, where $\lfloor a \rfloor$ is the largest integer smaller than $a$. This expression is based on the inclusion-exclusion principle for the probability of the union of events, for which the first term in the sum provides an upper bound and the first two terms provide a lower bound.

Averaging over $i$ (number of the nodes in the transmission range of $x$) and over the number of network nodes, we have:
\begin{align}
\label{eq:interiorupperbound}
\sigma'(N) &\leq \sum_{k=1}^{\infty} e^{-N}\frac{N^k}{k!} \nonumber \\
&~ \cdot k \sum_{i=0}^{k-1} \binom{k-1}{i} d^i(1-d)^{k-1-i} U_i(\eta) \nonumber\\
%&\leq \frac{e^{-N}}{1-e^{-N}}\sum_{k=1}^{\infty} \frac{N^k}{k!} \frac{a_1k}{k} \sum_{j=1}^{k} {k-1 \choose j-1} a_1^{j-1}(1-a_1)^{k-j} \nonumber \\
%& \left[(1-d)^{k-1}+ d(k-1) \sum_{i=0}^{k-1} {k-2 \choose i} (d(1-\eta))^{i-1}(1-d)^{k-1-i}\right] \nonumber\\
%&= \frac{e^{-N}}{1-e^{-N}} \left[N\sum_{k=1}^{\infty}\frac{(N(1-d))^{k-1}}{(k-1)!} + dN^2\sum_{k=2}^{\infty}\frac{(N(1-\eta d))^{k-2}}{(k-2)!}\right] \nonumber\\
%&\leq \frac{1}{1-e^{-N}}\left(Ne^{-dN} + dN^2 e^{-\eta dN}\right) \nonumber \\
&\leq dN^2e^{-\eta dN}\left(1+\frac{1}{dN}e^{-(1-\eta)dN}\right)\,.
\end{align}

\subsection{Calculation of $\sigma''(N)$}
\label{subsec:calculationofsigma2}

So far we have considered the interior nodes that are at least $R$-distance away from the boundary of the network region. Now, we consider edge nodes that are within $R$ of the network edge. Some $\eta$disks of an edge node may fall partially (up to half) outside the region, which increases the chance that they are empty. We refer to this phenomenon as the \emph{edge effect}. Since the network region is circular, the number of such edge nodes equals $(2-\sqrt{d})\sqrt{d}N$, which is of order $\Theta\left(\sqrt{d}N\right)$. We need to determine how their contribution to $\sigma(N)$ differs from the interior nodes.

Consider an edge node $e$, $(\delta' R)$-distance away from the network edge, with $0 < \delta' < 1$. As shown in Fig. \ref{fig:area}, we take node $e$ as the pole and the ray $eu$ (perpendicular to the network edge) as the polar axis of the \emph{local} (polar) coordinates at node $e$. We argued at the beginning of this section that, for edge node $e$, the probability of an $\eta$disk being empty, depends highly on its orientation. Let us consider this claim more closely. Let $\varphi := \cos^{-1}(\delta)$, as shown in Fig. \ref{fig:area}, where $\delta R$ is the distance between node $e$ and the line passing through the intersection points $B$ and $F$ in Fig. \ref{fig:area} with
$$\delta = \delta' - \frac{R}{L}\frac{1-\delta'^2}{2(1-\delta' \frac{R}{L})}\,,$$
and $L := \sqrt{|A|/\pi} = R/\sqrt{d}$ being the network region radius.
\begin{figure}
  \centering
  \includegraphics[width=2.8 in]{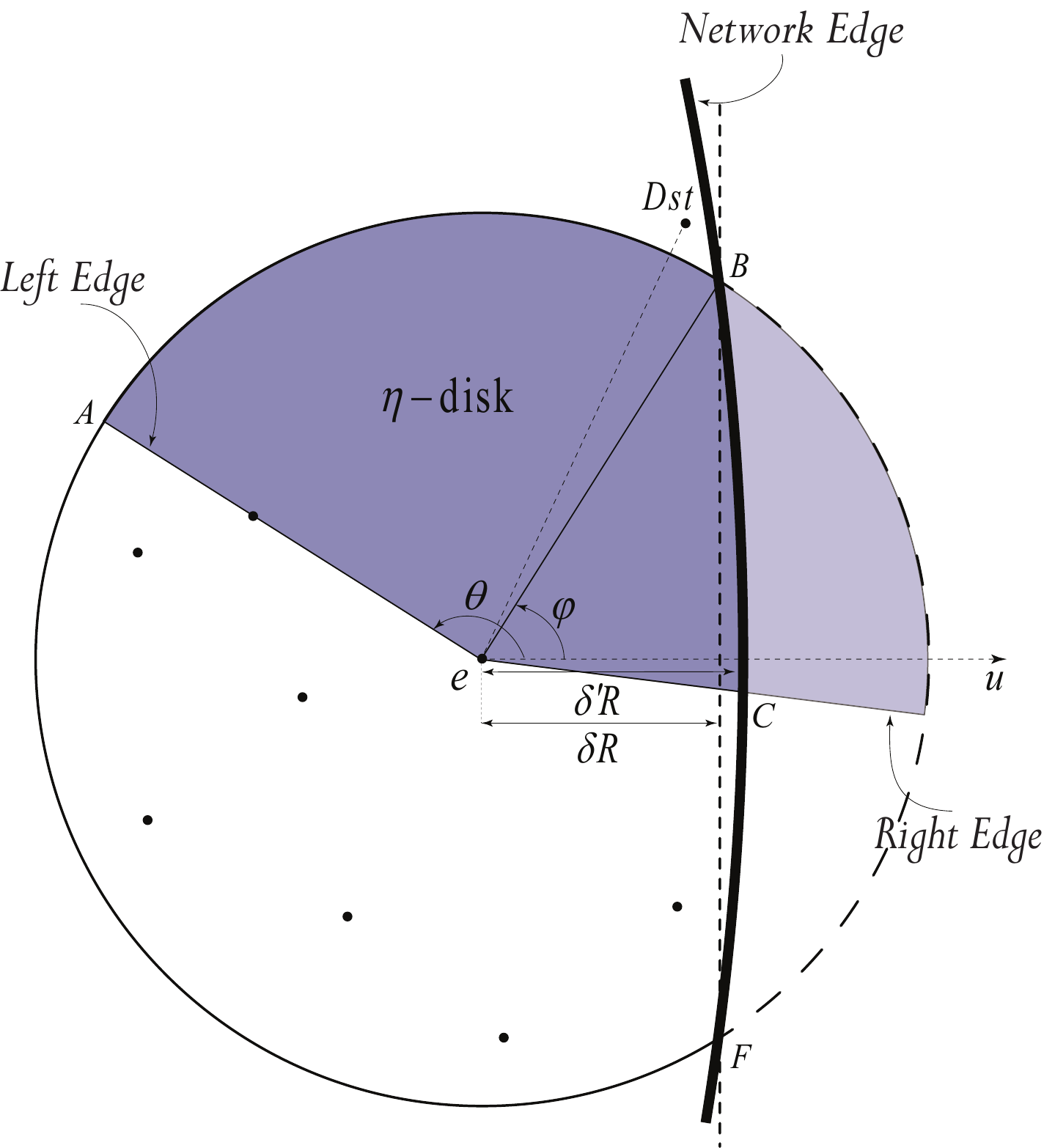}
  \caption{Intersection of the $\eta$disk with the network region.}
  \label{fig:area}
\end{figure}

Note that all the $\eta$disks are oriented towards the destination node. Hence, for all $\eta$disks that are oriented at an angle in the range $(-\varphi,\varphi)$, we must have that the destination is within node $e$'s transmission range. Therefore, we only need to be concerned with empty $\eta$disks oriented at an angle in the range $(\varphi,2\pi-\varphi)$. The $\eta$disks oriented at an angle in the range $(-\varphi-\eta \pi,-\varphi) \cup (\varphi,\varphi+\eta \pi)$ are partially outside the network region, as illustrated in Fig. \ref{fig:area}, and those oriented at any angle in $(\varphi+\eta \pi,2\pi-\varphi-\eta \pi)$ are fully contained inside the network region. Note that here, all the angles are measured relative to the polar axis $eu$. In both aforementioned cases, the area of the $\eta$disk inside the network region is at least $\eta \pi R^2/2$. Hence, we can compute a simple upper bound on $\sigma''(N)$ as follows. Let $a_2 := \pi(L^2 - (L-R)^2)/|A| = \sqrt{d}(2-\sqrt{d})$ and $a_1 := \pi(L^2-(L-2R)^2)/|A| = 4\sqrt{d}(1-\sqrt{d})$ be the normalized areas of the network edge region and the network extended edge region\footnote{The extended edge region is the area of the network that is within $2R$ of the network edge.} respectively. We have
\begin{align}
\label{eq:simplesigma''}
\sigma''(N) &\leq \sum_{l=1}^{\infty} e^{-a_1 N}\frac{(a_1 N)^l}{l!}\, l\, \frac{a_2}{a_1} \nonumber \\
&~ \cdot \sum_{i=0}^{l-1} {l-1\choose i} (\frac{d}{2a_1})^i(1-\frac{d}{2a_1})^{l-1-i} U_i(\frac{\eta}{2}) \nonumber\\
&\leq (2-\sqrt{d})\sqrt{d}Ne^{-\frac{1}{2}dN} + (1-\frac{\sqrt{d}}{2})d^{3/2}N^2e^{-\frac{\eta}{4}dN}\,.
\end{align}

A much tighter upper bound on $\sigma''(N)$ can be obtained as follows. First, suppose that there are \emph{no} nodes within the transmission range of node $e$; this event occurs with probability no greater than
\begin{align}
\label{eq:sigma''(no node)}
\sigma''(N) &\leq \sum_{l=1}^{\infty} e^{-a_1N}\frac{(a_1N)^l}{l!}\,l\,\frac{a_2}{a_1}(1-\frac{d}{2a_2})^{l-1} \nonumber\\
%&\leq \frac{e^{-a_2N}}{1-e^{-a_2N}} \sum_{l=1}^{\infty} 2 (1-\frac{d}{2a_2})^{l-1} l \frac{a_3}{a_2} \sum_{j=1}^{l} {l-1 \choose j-1} (\frac{a_3}{a_2})^{j-1}(1-\frac{a_3}{a_2})^{l-j} \nonumber \\
%&= \frac{2\frac{a_2}{a_3}a_2N}{1-e^{-a_2N}}e^{-a_2N} \sum_{l=1}^{\infty} \frac{(a_2N(1-\frac{d}{2a_2}))^{l-1}}{(l-1)!} \nonumber \\
&\leq (2-\sqrt{d})\sqrt{d}Ne^{-\frac{1}{2}dN}\,.
\end{align}

Second, suppose that there are $i \geq 1$ nodes in the intersection of node $e$'s transmission range with the network region. If an empty $\eta$disk exists and it is completely contained within the network region, W.L.O.G., there should be a node on its left edge at some angle $\theta \in (\varphi+2\eta\pi,2\pi-\varphi)$. However, for an empty $\eta$disk that is partially contained within the network region there should be, again W.L.O.G., a node at an angle $\theta \in (\varphi+\eta\pi,\varphi+2\eta\pi)$ or $\theta \in (-\varphi,-\varphi+\eta\pi)$ on the left edge of the $\eta$disk (note that, as discussed earlier, no $\eta$disks can be oriented at an angle in $(-\varphi,\varphi)$). Clearly, the existence probability of such empty $\eta$disks (that is partially contained in the region $A$) increases as either $\delta$ or $|\theta|$ decreases. The area of the intersection between such an $\eta$disk (that is partially contained in the region $A$) and the network region $A$ is that of a wedge with angle $|\theta|-\varphi$ (wedge $AeB$ in Fig. \ref{fig:area}) plus at least a triangle abutting the right edge of the wedge (triangle $BeC$ in Fig. \ref{fig:area}). % with area no more than $\frac{R^2}{2}\cot(\varphi)$.
In fact for an arbitrary small $\epsilon$, if either $\delta \geq \sin(3\epsilon\pi)$ or $\theta \geq \varphi + \eta\pi + 2\epsilon\pi$, the area of the intersection between the $\eta$disk and the network region is at least $(\eta/2+\epsilon)\pi R^2$. Otherwise, it is at least $\eta\pi R^2/2$. Thus, averaging over $\delta$, $\theta$ and the number of edge nodes, the probability that some edge nodes have empty $\eta$disks in some directions, $\sigma''(N)$, is derived to be no more than
%{\allowdisplaybreaks
\begin{align}
\label{eq:detailedgeupperbound}
&\sum_{l=1}^{\infty} e^{-a_2N}\frac{(a_2N)^l}{l!}\,l\, \frac{a_2}{a_1}\sum_{i=1}^{l-1} {l-1\choose i} (\frac{d}{2a_1})^i(1-\frac{d}{2a_1})^{l-1-i} i \nonumber\\
&~ \cdot \bigg\{\p{\delta < \sin(3\pi\epsilon)} \p{\exists \textrm{ empty $\eta$disk} ~\Big|~ i, \delta < \sin(3\pi\epsilon)} \nonumber\\
&~ + \p{\delta > \sin(3\pi\epsilon)} \p{\exists \textrm{ empty $\eta$disk} ~\Big|~ i, \delta > \sin(3\pi\epsilon)} \bigg\} \nonumber\\
%&\leq \frac{\frac{2a_3}{a_2}e^{-a_2N}}{1-e^{-a_2N}}\sum_{l=1}^{\infty} \frac{(a_2N)^l}{(l-1)!}  \sum_{i=1}^{l-1} {l-1\choose i} (\frac{d}{2a_2})^i(1-\frac{d}{2a_2})^{l-1-i}  \nonumber\\
%&~ \cdot i \bigg\{\frac{3\pi\epsilon}{1-\frac{\sqrt{d}}{2}}\bigg[4\epsilon(1-\frac{\eta}{1+8\epsilon})^{i-1} + 2\eta (1-\frac{\eta+2\epsilon}{1+8\epsilon})^{i-1}\nonumber \\
%&~+ (1-2\eta)(1-\frac{2\eta}{1+8\epsilon})^{i-1}\bigg] \nonumber\\
%&~+ \bigg[2\eta(1-(\eta/2+\epsilon))^{i-1} + (1-2\eta)(1-\eta)^{i-1}\bigg]\bigg\} \nonumber\\
&\leq \frac{d^{3/2}N^2}{2-\sqrt{d}}\bigg\{12\pi\epsilon^2e^{-\frac{\eta dN}{1+8\epsilon}}+6\pi\epsilon e^{-\frac{(\eta+2\epsilon) dN}{1+8\epsilon}}+ 3\pi\epsilon e^{-\frac{2\eta dN}{1+8\epsilon}} \nonumber\\
&~ + 2e^{-\frac{(\eta+2\epsilon) dN}{2}} + e^{-\eta dN}\bigg\}\,,
\end{align}%}
for arbitrary $\epsilon \geq 0$. Choosing $\epsilon = \frac{2\log dN}{dN}$, together with \eqref{eq:sigma''(no node)}, yields a tighter upper bound for the probability that some edge nodes has an empty $\eta$disk oriented in some direction:
\begin{align}
\label{eq:edgeupperbound}
\sigma''(N) &\leq \frac{400\pi\left(\log dN\right)^2}{\sqrt{d}}e^{-\frac{\eta}{2}dN}\nonumber \\
&~+ \frac{16(dN)^2}{\sqrt{d}}e^{-\eta dN} + 4\sqrt{d}Ne^{-\frac{1}{2}dN}\,,
\end{align}
for large enough $dN$ where the last summand is the probability that some edge nodes have no other nodes within their transmission ranges, derived in \eqref{eq:sigma''(no node)}.

\subsection{Calculation of $\sigma(N)$}
\label{subsec:calculationofsigma}

Finally, summing \eqref{eq:interiorupperbound} and \eqref{eq:edgeupperbound}, we obtain the bound $\sigma(N)$ on the probability that some nodes in the network have empty $\eta$disks looking in some directions as:
\begin{align}
\label{eq:finalupperbound}
\sigma(N) &~\leq~ \frac{400\pi\left(\log dN\right)^2}{\sqrt{d}}e^{-\frac{\eta}{2}dN}+\frac{16(dN)^2}{\sqrt{d}}e^{-\eta dN} \nonumber \\
&~+ 4\sqrt{d}Ne^{-\frac{1}{2}dN} + 4dN^2e^{-\eta dN}\,.
\end{align}

This bound on $\sigma(N)$ is asymptotic to $\frac{400\pi\left(\log dN\right)^2}{\sqrt{d}}e^{-\frac{\eta}{2}dN}$, which goes to zero if $\eta dN + \log d \to \infty$ as $N\to\infty$. Hence, setting $d = \frac{c\log N}{N}$ with $c > 1/\eta$, we obtain that every node in the network have at least one relaying node in every direction over which their targeted destinations may lie with probability approaching one as $N \to \infty$, which shows the consistency between our result and the ones derived in \cite{GuptaKumar2}, \cite{conn} and \cite{yin} for $\eta = 1$.

\begin{remark}
\label{rem:d}
Setting $d = \frac{c\log N}{N}$ is equivalent to setting $R(\lambda,|A|) = \sqrt{\frac{c}{\pi}\frac{\log\lambda + \log |A|}{\lambda}}$ for $c > 1/\eta$. In particular, for the case of dense networks (i.e., $\lambda\to\infty$ with a finite $|A|$) and for the case of large-scale networks (i.e., $|A|\to\infty$ with a finite $\lambda$), setting $R(\lambda) = K \sqrt{\log\lambda/\lambda}$ and $R(|A|) = K \sqrt{\log |A|}$ respectively, with a large enough constant $K$, guarantees the existence of relaying nodes in a ``uniform'' manner around each node in the network.
\end{remark}

\section{Theorem \ref{thm1}.$ii$--$v$ Proof: Path Length Statistics and Connectivity}
\label{sec:Path Length Statistics}

Assume $R(N)$ is chosen such that $\eta dN + \log d \to +\infty$ as $N\to\infty$ and $N$ is large enough such that each node in the network has at least one relaying node in every direction with high probability. We now investigate the question of how long the path generated by the random $\eta$disk routing scheme is, where we focus on the $\eta = 1/2$ case in this paper. To answer this question, we need to characterize the process of path establishment from a given source to its destination by the random $\frac{1}{2}$disk routing scheme.

In the following, we ignore the edge effect for the sake of simplicity in mathematical characterization. In other words, we assume that the $\frac{1}{2}$disks of \emph{all} network nodes looking in any direction are completely contained in the network region. Later, we show (through simulation) in Section \ref{sec:Simulation Results} that the asymptotic results derived in this section still hold even when considering the routing next to the boundary for source-destination pairs that are located near the network boundary.

Now consider an arbitrary source-destination pair that is $h$-distance apart. We set the destination node at the origin and assume that the routing path starts from the source node at $X_0 = (-h,0)$, where $X_n$ is the Cartesian coordinate of the $n^{\textrm{th}}$ relay node along the routing path and $r_n:=\|X_n\|$ is the Euclidean distance of the $n^{\textrm{th}}$ relay node from the destination.

More specifically, the routing path starts at the source node $X_0 = (-h,0)$ with its $\frac{1}{2}$RSR $D_0$ that is a $\frac{1}{2}$disk with radius $R$ centered at $X_0$ and oriented towards the destination at $(0,0)$. The next relay $X_1$ is selected at random from those contained in $D_0$ (the rRSL rule). This induces a new $\frac{1}{2}$RSR $D_1$, also a $\frac{1}{2}$disk but centered at $X_1$ and oriented towards the destination. Relay $X_2$ is selected randomly among the nodes in $D_1$, and the process continues in the same manner until the destination is within the transmission range. Note that $D_n$ solely depends on $X_n$. We claim that the routing path has converged (or is established) whenever it enters the transmission/reception range of the final destination, i.e., $r_{\nu} \leq R$, for some $\nu \in \{1, 2, \cdots\}$. In Fig. \ref{fig:routingevolution}, we illustrate the progress of routing towards the destination.
\begin{figure}
  \centering
  \includegraphics[width=3 in]{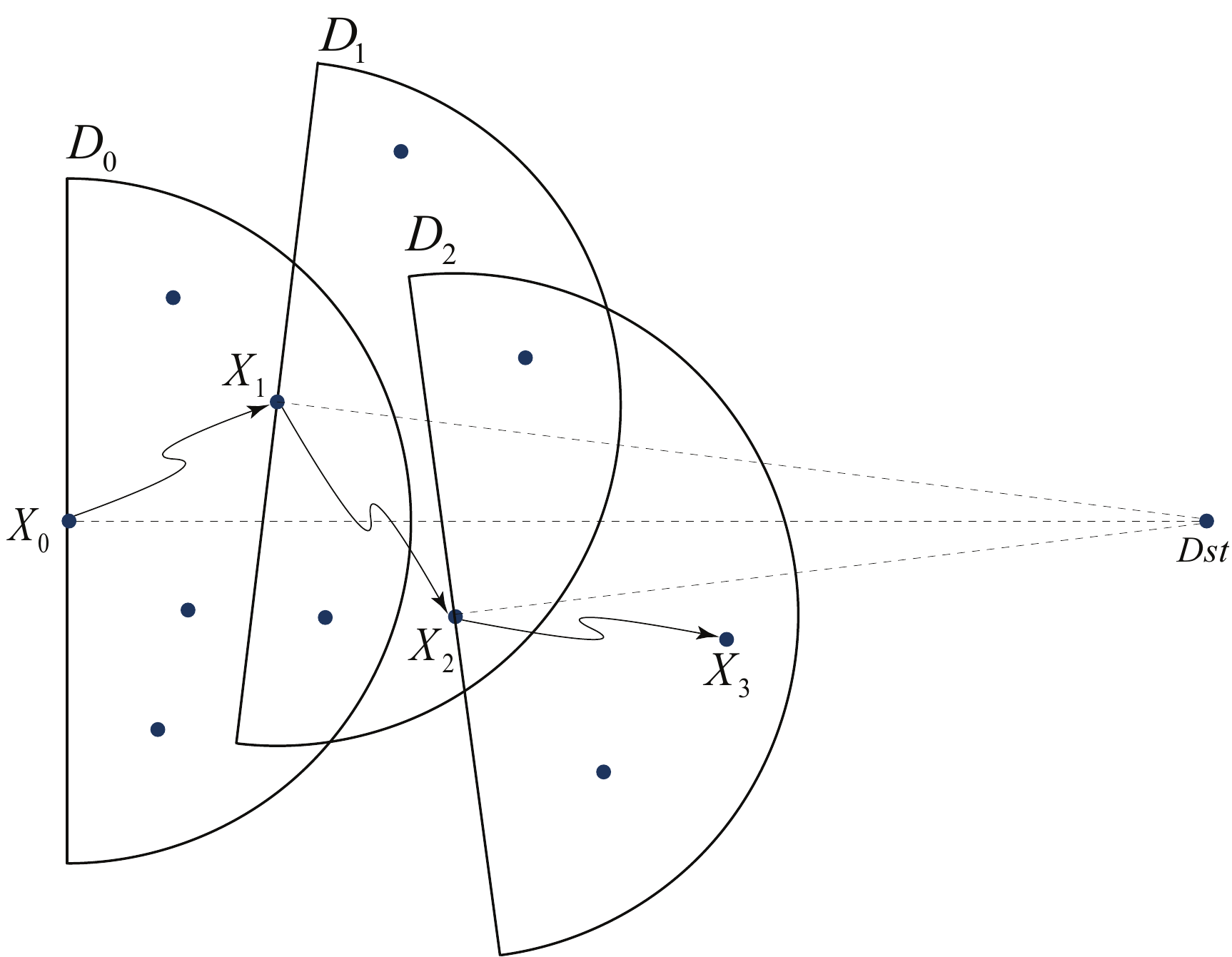}
  \caption{Evolution of the random $\frac{1}{2}$disk routing path.}
  \label{fig:routingevolution}
\end{figure}

Define $S_n := h - r_n$ and the routing increment as $Y_n := S_n - S_{n-1} = r_{n-1} - r_n$. Let $\phi(D_n)$ be the number of nodes in $D_n$. For the sake of definition, we set $Y_i = 0$ for $i>n$ if $\phi(D_{n-1}) = 0$. In the next subsection we investigate how similar $\{S_n\}$ and consequently $\{r_n\}$ are to a Markov process.\footnote{For an alternative treatment of the problem refer to \cite{similar1}, Section 4.1.}

\subsection{Theorem \ref{thm1}.$ii$ Proof: Markov Approximation}
\label{subsec:Markov Approximation}

In this subsection we investigate how close our Markov approximation model for $\{r_n\}$ is to the actual process of route establishment by the random $\frac{1}{2}$disk routing scheme. Observe that even though the underlying distribution of the network nodes is Poisson and the new relays are chosen uniformly at random within each $\frac{1}{2}$RSR, the increments $Y_1, Y_2, \ldots$ are neither independent nor identically distributed. This is due to the fact that the orientations of all $\frac{1}{2}$RSRs are pointing to a common node (destination) and might overlap, as shown in Fig. \ref{fig:routingevolution}.

More specifically, let $k_n$ be the number of previous relaying nodes whose RSRs intersect with $D_n$. Assuming $\phi(D_{n}) > 0$, $S_{n+1} = S_n + Y_{n+1}$ is a Markov process if the conditional distribution of $Y_{n+1}$ given $S_i$, $n-k_n \leq i\leq n$, only depends on $S_n$. Equivalently, $r_{n+1} = h - S_{n+1}$ is a Markov process if the conditional distribution of $X_{n+1}$ given $X_i$, $n-k_n \leq i\leq n$, only depends on $X_n$. However, the overlap of $D_n$ with $D_j$, $n-k_n\leq j < n$, correlates the spatial distribution of nodes in $D_n$ (and consequently $X_{n+1}$ and $Y_{n+1}$), not only with $X_n$, but also possibly with $X_j$, $n-k_n \leq j<n$.\footnote{This dependence increases as the packet gets closer to the destination due to the fact that the overlapping area between $D_n$ and $D_{n-1},$ $D_{n-2},$ $\ldots$ increases (stochastically) as the packet gets closer to the destination. In \cite{similar1} and its companion papers \cite{similar2}--\cite{similar4}, the authors looked at hop length distributions in ad hoc sensor networks with geometric routing schemes, and reported similar dependencies between hop increments $Y_1, Y_2, \ldots$.} In fact, given $X_i$, $n-k_n \leq i \leq n$, the nodes contained in $D_n$ are no longer uniformly distributed over $D_n$ as one would expect for a Poisson distributed network due to the overlap of $D_n$ with $D_j$, $n-k_n\leq j<n$ (cf. Proposition \ref{prop1}). As such, the process of path establishment by the random $\frac{1}{2}$disk routing scheme, $\{r_n\}$, is \emph{not} a Markov process. What is less clear, however, is how close $\{r_n\}$ is to a Markov process.

Tracking the dependence of $X_{n+1}$ on all $X_j$, $n-k_n \leq j \leq n$, is extremely tedious. As such, in this work we only show how close the routing path progress conditioned on the previous two hops is to a Markov process, i.e., we show in Proposition \ref{prop1} that the conditional distribution of $X_{n+1}$ given $(X_n,X_{n-1})$ is close to that of $X_{n+1}$ given $X_n$ for large $N$. We show that the error resulted from considering only $X_n$ and neglecting the effect of $X_{n-1}$ on the distribution of $X_{n+1}$ is at most $1/(dN)$, which goes to zero as $N\to\infty$.\footnote{Note that by Theorem \ref{thm1}, $R$ is chosen such that $\eta dN+\log d \to +\infty$ as $N\to\infty$, which implies that $dN\to \infty$ and $d\to 0$ for smallest transmission radius \cite{GuptaKumar}.}

Note that, by a method similar to the proof of Proposition \ref{prop1}, we might show that the incurred errors in modeling $\{r_n\}$ due to higher-order dependencies should be at most $k_n/(dN)$, which is relatively negligible \emph{if} $k_n = o\left(\sqrt{dN}\right)$ for large $N$. Simulations indicate that $k_n$ should in fact remain in the order of $o\left(\sqrt{dN}\right)$; however, we could not establish an explicit proof for this claim, which will be left for our future study.

We emphasize that, in what follows, conditioning on $\phi(D_n) > 0$ means we only know that there is at least one node in $D_n$; however, conditioning on $\phi(D_n)$ means we know the exact number of nodes in $D_n$. Furthermore, Let $C^c:=A - C$ denote the complement of $C$ with respect to network region $A$ and $\I{\cdot}$ represent the indicator function, i.e., $\I{\cdot} = 1$ if the event in the subscript happens and $\I{\cdot} = 0$ otherwise.

Now we investigate how similar the distribution of $X_{n+1}$ over $D_n$ is to a uniform distribution given $(X_n,X_{n-1})$. Note that given only $X_n$, $X_{n+1}$ is uniformly distributed over $D_n$. Given $X_n$, $X_{n-1}$, $\phi(D_{n-1})$, and $\phi(D_n) > 0$, the number of nodes in $D_{n-1}D_n := D_{n-1} \cap D_{n}$ is $\phi(D_{n-1}D_n) \sim \textrm{Binomial}\left(\phi(D_{n-1})-1,\frac{|D_{n-1}D_n|}{|D_{n-1}|}\right) + \I{X_{n-1} \in D_n}$ and is independent of the number of nodes in $D^c_{n-1}D_n$, which is $\phi(D^c_{n-1}D_n) \sim \textrm{Poisson}(\lambda |D^c_{n-1}D_n|)$. Moreover, conditioned additionally on the two random variables $\phi(D_{n-1}D_n)$ and $\phi(D^c_{n-1}D_n)$, each collection of nodes (located in $D_{n-1}D_n$ and $D^c_{n-1}D_n$) is uniformly distributed over the respective areas. This does not, however, imply that the combined collection of nodes is uniformly distributed over $D_n$ as shown in the following proposition. The combined points are uniformly distributed over $D_n$ \emph{only if} the (conditional) expected proportion of points in $D_{n-1}D_{n}$ is $\eee{\frac{\phi(D_{n-1}D_n )}{\phi(D_n)} \mid \phi(D_n)>0, \phi(D_{n-1})>0, X_n, X_{n-1}} = \frac{|D_{n-1}D_{n}|}{|D_{n}|}$.

%Nonetheless, according to the following proposition, the error resulted from proceeding as if $X_{n+1}$ is located uniformly on $D_n$ is negligible for large $N$. Essentially, knowing $X_n$, the distribution of $X_{n+1}$ in $D_n$ is \emph{almost} uniform over $D_n$ and independent of the location of the previous relaying node $X_{n-1}$ for large $N$.

\begin{prop}
\label{prop1}
Assume the locations of current and previous relay nodes $(X_n,X_{n-1})$ are given and $\phi(D_{n-1})>0$. Given $\phi(D_n)>0$, the distribution of the nodes located inside $D_n$ converges to a uniform distribution over $D_n$ as $N \to \infty$. In particular, the conditional probability of selecting the next node $X_{n+1}$ from $D_{n-1}D_n$, i.e., $\rho(X_{n-1},X_n) := \eee{\
\frac{\phi(D_{n-1}D_n)}{\phi(D_n)}~\Big|~ \phi(D_n)>0, \phi(D_{n-1})>0, X_n, X_{n-1}}$ satisfies
\begin{align}
\label{eq:prop1}
\left(1-\frac{2}{dN} - \alpha_1(n) e^{-\alpha_2(n) dN}\right)\frac{|D_{n-1}D_n|}{|D_n|} &  \nonumber \\
< \rho(X_{n-1},X_n) <& \frac{|D_{n-1}D_n|}{|D_n|}\,,&
\end{align}
where $\alpha_1(n) > 2$ and $0 < \alpha_2(n) < 1$ are independent of $N$.
\end{prop}
\begin{proof}
Refer to Appendix \ref{appendix:markov}.
\end{proof}

Observe that according to \eqref{eq:prop1}, given the locations of two previous relay nodes $(X_{n-1},X_n)$, it is less likely that the next relay $X_{n+1}$ is selected from $D_{n-1}D_n$ as opposed to the case where the nodes were actually uniformly distributed over $D_n$. Hence, $X_{n+1}$ is \emph{not} uniformly distributed over $D_n$ given $(X_{n-1},X_n)$. However, we have $\rho(X_{n-1},X_n) \to \rho(X_n) = |D_n D_{n-1}|/|D_n|$ as $N\to\infty$. Hence, the routing path progress given the second-order history of the routing path converges asymptotically to a Markov process. Nevertheless, the routing increments $Y_1, Y_2, \ldots$ are not identically distributed and as shown in the next subsection, $Y_{n+1}$ is in fact a function of $r_n$. As such, in the following, we proceed \emph{as if} the process that governs the path establishment by the random $\frac{1}{2}$disk routing scheme is a \emph{non-homogeneous Markov process} for large $N$.

\subsection{Theorem \ref{thm1}.$iii$ and $v$ Proof: Expected Length of the Random $\frac{1}{2}$disk Routing Path and network Connectivity}
\label{subsec:Expected Length}

Using the Markovian approximation model for the routing path evolution $\{r_n\}$, we now derive the asymptotic statistics for the length of the random $\frac{1}{2}$disk routing paths. Let $X_n$ be the $n^{\textrm{th}}$ hop of the routing path and $(x'_{n+1},y'_{n+1})$ be the projection of $X_{n+1}-X_n$ onto the \emph{local} Cartesian coordinates with node $X_n$ as the origin and the $x$-axis pointing from $X_n$ to the destination node as shown in Fig. \ref{fig:routingschemebound}. Hence,
\begin{equation}
\label{eq:alternativeformulation}
r_{n+1} = \sqrt{(r_n-x'_{n+1})^2 + y'^{2}_{n+1}}\,,
\end{equation}
characterizes the distance evolution of the routing path at the $n^{\textrm{th}}$ hop. Based on the Markov approximation model, $X_{n+1}$ is uniformly distributed over $D_n$; hence $\{(x'_{n},y'_{n})\}$ is an i.i.d. sequence of random variables with ranges $x'_{n} \in [0,R]$ and $y'_{n} \in [-R, R]$ for all $n$.% (Note that according to Proposition \ref{prop1}, $X_{n+1}$ is almost uniformly distributed on $D_n$ for large $N$.)}

\begin{figure}
  \centering
  \includegraphics[width=3 in]{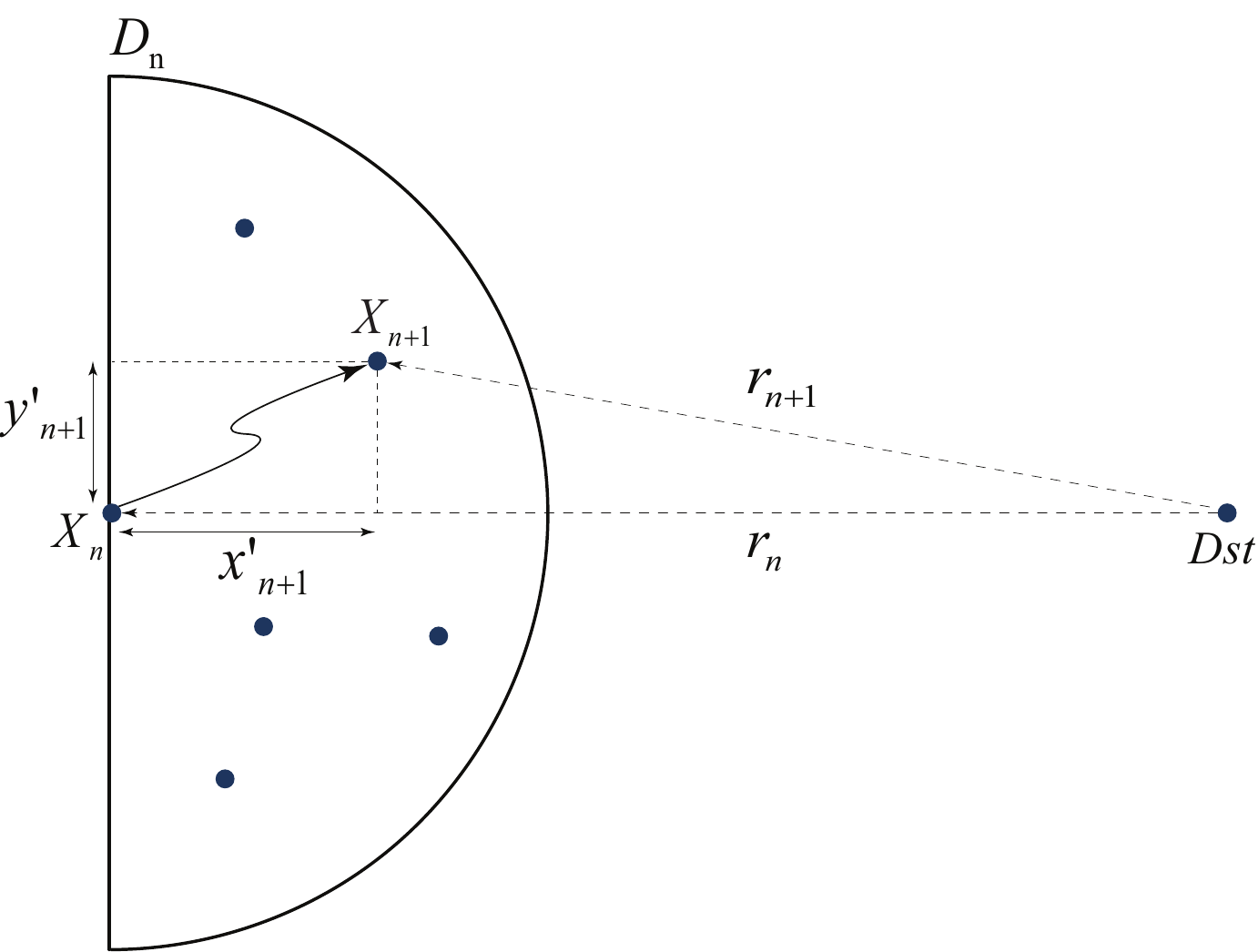}
  \caption{Distance between the next relay and the current node projected onto to the local coordinates at the current node.}
  \label{fig:routingschemebound}
\end{figure}

Define $\nu^{(h)}_{r} := \inf\{n:r_n \leq r,\, r_0 = h\}$, $r\geq R$, to be the index of the first relay node (along the routing path) that gets closer than $r$ to the destination when the source and destination nodes are $h$-distance apart. Hence, $\nu^{(h)}_{R}$ represents the first time the routing path enters the reception range of the destination and $\nu^{(h)}_{R}+1$ quantifies the length of the routing path, where $\nu^{(h)}_{R}\sim\nu^{(h)}_{R}+1$. Recall that in this paper we define the length of a routing path as the number of hops traversed over the routing path. It is easy to show that $\nu^{(h)}_{r}$ is a stopping time \cite{Grimmett} and
\begin{equation}
r-R \le r_{\nu^{(h)}_r}\le r\,. \nonumber
\end{equation}

Furthermore, let $g(r,x',y'):=\sqrt{(r-x')^2+y'^2}-r$. Observe that $g$ is a nonincreasing function over $r > R$, for fixed $(x',y')$, and $g(r_n,x'_{n+1},y'_{n+1}) = -Y_{n+1}$. Thus, for $n<\nu^{(h)}_r$, we have $r_n > r$ and
\begin{equation*}
-x'_{n+1} \le r_{n+1}-r_n = g(r_n,x'_{n+1},y'_{n+1})\le g(r,x'_{n+1},y'_{n+1})\,.
\end{equation*}

Hence, for a source-destination pair that is $h$-distance apart ($r_0 = h$), we have
\begin{subequations}
\begin{align}
r-R &\le r_{\nu^{(h)}_r}\le h+\sum_{n=1}^{\nu^{(h)}_r}g(r,x'_n,y'_n)\,, \label{eq:ba} \\
h+\sum_{n=1}^{\nu^{(h)}_r}(-x'_n) &\le r_{\nu^{(h)}_r}\le r\,.  \label{eq:bb}
\end{align}
\end{subequations}

Note, as well, that (refer to Appendix \ref{appendix:expecting})
\begin{align}
\label{eq:expect}
-\frac{4R}{3\pi}=\ee{-x'_n} \le &\ee{g(r,x'_n,y'_n)} \le \nonumber \\
&\ee{g(R,x'_n,y'_n)}<-\frac{R}4<0\,.
\end{align}

Now, applying Wald's equality \cite{Resnik} to \eqref{eq:ba} and \eqref{eq:bb} and rearranging, we obtain a bound on the expected value of the stopping time $\nu^{(h)}_{r}$:
\begin{align}
\label{eq:genralbound}
\frac{3\pi(h-r)}{4R}&\le \ee{\nu^{(h)}_r\mid h} \le\frac{h-r+R}{-\ee{g(r,x'_n,y'_n)}} \nonumber \\
&\le\frac{h}{-\ee{g(r,x'_n,y'_n)}}\le\frac{4h}R\,.
\end{align}
Substituting $r$ with $R$ we obtain a general bound for the expected length of routing path (minus one) between a source-destination pair that is $h$ distance apart as
\begin{equation*}
\frac{3\pi}{4}\left(\frac{h}{R}-1\right)\le \ee{\nu^{(h)}_R\mid h} \le\frac{4h}R\,.
\end{equation*}
This implies that the length of the random $\frac{1}{2}$disk routing path is almost surely (a.s.) finite when each network node has at least one node in its $\frac{1}{2}$RSR looking in any direction, which happens with probability no less than $1-\sigma(N)$ as obtained in \eqref{eq:finalupperbound}. In other words, when $dN/2 + \log d \to\infty$ as $N\to\infty$, we obtain that $\pp{\nu^{(h)}_{R} < \infty} \to 1$ as $N\to\infty$. This in turn shows that given $dN/2 + \log d \to\infty$ as $N\to\infty$, every path starting from any source will reach its destination in finitely many hops a.a.s., which proves that the network is connected employing the random $\frac{1}{2}$disk routing scheme, according to the connectivity definition in Section \ref{sec:SystemModel}.

When the ratio $h/R$ (i.e., the ratio between the source-destination distance and the transmission range) is large, we can obtain a tighter bound on the expected length of the routing path between a source-destination pair with $h$ separation. For the following, we assume $h \geq 2R$. Since $r_{\nu^{(h)}_r}\leq r$, we must have $\eee{\nu^{(h)}_R\mid h}\le \eee{\nu^{(h)}_r\mid h}+\eee{\nu^{(r)}_R\mid r}$. Thus, by \eqref{eq:genralbound} and proper substitutions, we have
\begin{equation*}
\frac{3\pi}{4}\left(1 - \frac{R}{h}\right)\le \frac{\ee{\nu^{(h)}_R\mid h}}{h/R} \le \frac{R}{-\ee{g(r,x'_n,y'_n)}}+\frac{4r}{h}\,,
\end{equation*}
for all $R \leq r \leq h$. Using
\begin{equation}
\label{eq:bound on g}
-x'_n \leq g(r,x'_n,y'_n) \leq - x'_n + \frac{(y'_n)^2}{2(r-R)} \,,
\end{equation}
and \eqref{eq:second moment} we get $\eee{g(r,x'_n,y'_n)} \leq -\frac{4R}{3\pi} + \frac{R^2}{8(r-R)}$. Choosing $r$ such that $\frac{8(r-R)}{R} = \frac{3\pi}{4}(\sqrt{\frac{h}{2R}}+1)$ (we may do so using the intermediate value theorem and the fact that $R \leq r \leq h$ and $h \geq 2R$), we may determine that
%%%%%%%%%%%%%%%%%%%%%%%%%%%%%%%%%%%%%%%%%%%%%%%%%%%%%%%%%%%%%%%%%%%%%%%%%%%%%%%%%%%%%%%%%%%%%
%% Proof of intermediate value Theorem:
%\begin{align*}
%\frac{3\pi}{4}\left(\sqrt{\frac{h}{2R}}+1\right) \leq \frac{3\pi}{2}\sqrt{\frac{h}{2R}} &\leq^{?} \frac{8(h-R)}{R}\,, \\
%\frac{3\pi}{16}\sqrt{\frac{h}{2R}}+1 &\leq^{?} \frac{h}{R}\,, \\
%\frac{3\pi}{16}\sqrt{\frac{h}{2R}}+1 \leq (\frac{3\pi}{16}+1)\sqrt{\frac{h}{2R}} &\leq^{?} \frac{h}{R}\,, \\
%\frac{\frac{3\pi}{16}+1}{\sqrt{2}} \leq^{\checkmark} \sqrt{\frac{h}{R}}\,.
%\end{align*}
%also $\frac{r}{R} = \frac{3\pi}{32}((\frac{h}{2R})^{1/2}+1)+1$
%%%%%%%%%%%%%%%%%%%%%%%%%%%%%%%%%%%%%%%%%%%%%%%%%%%%%%%%%%%%%%%%%%%%%%%%%%%%%%%%%%%%%%%%%%%%%
\begin{align}
\label{eq:E(nu)-bound}
\frac{3\pi}{4}\left(1 - \frac{R}{h}\right)&\le \frac{\ee{\nu^{(h)}_R\mid h}}{h/R} \le \frac{3\pi}{4}\frac{1}{1-\left(\sqrt{\frac{h}{2R}}+1\right)^{-1}} \nonumber \\
&~+ \frac{4R}{h}\left(\frac{3\pi}{32}\left(\sqrt{\frac{h}{2R}}+1\right)+1\right) \nonumber\\
&= \frac{3\pi}{4} \left(1+ \frac{5}{2}\sqrt{\frac{R}{2h}}+\frac{R}{2h}\right)+\frac{4R}{h}\,.
\end{align}
This implies
\begin{equation}
\label{eq:expected length_1}
\frac{R}{h}\ee{\nu^{(h)}_R \mid h} \to \frac{R}{\ee{x'_n}} = \frac{3\pi}{4}\,,
\end{equation}
or
\begin{equation}
\label{eq:expected length_2}
\ee{\nu^{(h)}_R \mid h} \sim \frac{h}{\ee{x'_n}} = \frac{3\pi}{4}\frac{h}{R}\,,
\end{equation}
as $\frac{h}{R} \to \infty$ given that $r_0 = h$.

\begin{remark}
\label{rem:h}
Recall that $L = \sqrt{|A|/\pi}$ and observe that $\p{h \leq \alpha} \leq \frac{\pi \alpha^2}{|A|}$. Therefore, we can obtain that $\p{h \leq \alpha(N)} \to 0$ for $\alpha(N) = o\left(L\right)$ as $N\to\infty$, which in return implies that $\p{h/R\to\infty \mid \eta dN+\log d \to\infty \textrm{ as } N\to\infty} = 1$. Hence, assuming that the conditions in Theorem \ref{thm1}.i hold, we have $h/R\to\infty$ a.s. as $N\to\infty$.
\end{remark}

\begin{remark}
\label{rem:asymtotic}
The asymptotic expected length of the routing path established by the random $\frac{1}{2}$disk routing scheme is $\frac{3\pi}{4} = R/\ee{x'_n} \approx 2.36$ times greater compared to the length of the routing path generated by the ideal direct-line routing scheme; in the ideal direct-line routing scheme we assume that there are relays located on the line connecting the source and destination with the maximal separation $R$.
\end{remark}

By averaging over all the possible source-destination pair distances $h$, we can determine the expected length of a typical random $\frac{1}{2}$disk routing path. Again using $\p{h\leq \alpha R} \leq \frac{\pi}{|A|}(\alpha R)^2$ and \eqref{eq:E(nu)-bound} we have that
\begin{align*}
\ee{\nu_R} &= \ee{\ee{\nu^{(h)}_R\mid h}\I{h\leq \alpha R} + \ee{\nu^{(h)}_R\mid h}\I{h> \alpha R}} \\
&\leq \frac{\pi \alpha^3R^2}{|A|}\left[\frac{3\pi}{4}\left(1+\frac{5}{\sqrt{8\alpha}}+\frac{1}{2\alpha}\right)+\frac{4}{\alpha}\right]\\
&~+ \frac{3\pi}{4}\frac{\ee{h\I{h > \alpha R}}}{R}\left(1+\frac{5}{\sqrt{8\alpha}}+\frac{1}{2\alpha}\right)+4\,,
\end{align*}
and
\begin{align*}
\ee{\nu_R} = \ee{\ee{\nu^{(h)}_R\mid h}} \geq \frac{3\pi}{4}\left(\frac{\ee{h}}{R}-1\right)\,.
\end{align*}

The problem of quantifying $\ee{h}$ is well studied in the literature \cite{averagedistance}, with the following known results for two network regions specifically: If the region is a planar disc with diameter $2L$, we have $\ee{h} = 128L/(45\pi) \approx 0.9054L$; and if it is a square with side length $L$, we have $\ee{h} = \left(2+\sqrt{2}+5\log(\sqrt{2}+1)\right)L/15 \approx 0.5214L$. Choosing $\alpha = o\left(d^{-1/6}\right)$ and recalling Remark \ref{rem:h}, we observe that $\ee{h\I{h > \alpha R}}\to \ee{h}$ as $N\to\infty$; hence, we have
%%%%%%%%%%%%%%%%%%%%%%%%%%%%%%%%%%%%%%%%%%%%%%%%%%%%%%%%%%%%%%%%%%%
%Extra Stuff:
%since $h=O\left(L\right)=O\left(d^{-1/2}R\right)$, then for any $\alpha = o\left(d^{-1/2}\right)$ we have that $\p{h>\alpha R}=1$. One choice of such $\alpha$ is $d^{-1/6}$.
%%%%%%%%%%%%%%%%%%%%%%%%%%%%%%%%%%%%%%%%%%%%%%%%%%%%%%%%%%%%%%%%%%%
\begin{align}
\label{eq:typical expected length}
\ee{\nu_R} \sim \frac{32}{15}\frac{1}{\sqrt{d}}\,,
\end{align}
as $N\to\infty$.

\subsection{Theorem \ref{thm1}.$iv$ Proof: Variance of the Random $\frac{1}{2}$disk Routing Path Length}
\label{subsec:Length Variance}

So far we have characterized the expected length of the routing paths generated by the random $\frac{1}{2}$disk routing scheme. However, the expected value alone is not descriptive enough regarding the \emph{individual realizations} of the routing path length. We need to determine how much the individual realization deviates from the expected value. Therefore, in this section, we consider the variance of the path lengths generated by the random $\frac{1}{2}$disk routing scheme. We first show that the variance is finite almost surely and then we show that asymptotically it grows linearly with the expected path length:

\begin{align}
\label{eq:bound on Var(R)}
\frac{\vv{\nu^{(h)}_R\mid h}}{\ee{\nu^{(h)}_R\mid h}} \to \frac{\vv{x'_n}}{\left(\ee{x'_n}\right)^2} = \frac{9\pi^2}{64}-1,
\end{align}
as $N\to\infty$. We will frequently use the following well known inequalities
\begin{equation*}
\left|\sqrt{\ee{X^2}}-\sqrt{\ee{Y^2}}\right| \leq \sqrt{\ee{(X-Y)^2}}\,,
\end{equation*}
and
\begin{equation*}
\left|\sqrt{\vv{X}}-\sqrt{\vv{Y}}\right| \leq \sqrt{\vv{X-Y}}\,.
\end{equation*}

Consider the i.i.d. sequence $\{(x'_n,y'_n)\}$ as defined in Section \ref{subsec:Expected Length}, and define the generalized stopping time $\nu^{(b)}_a$ to be $\nu^{(b)}_{a} := \inf\{n:r_n \leq a,\, r_0 = b\}$ for $R \leq a < b \leq h$. Observe that $\{\nu^{(b)}_a \geq N\}$ and $\{x'_n\}_{n < N}$ are independent, and $\eee{\nu^{(b)}_{a}}<\infty$ and $\ee{(x'_n)^2}<\infty$ as shown in Section \ref{subsec:Expected Length} and Appendix \ref{appendix:expecting}.

Note first that, by definition,
\begin{equation*}
\ee{\sum_{i = 1}^{\nu^{(h)}_R \wedge n} (-g(r_{i-1},x'_i,y'_i))} = \ee{r_0 - r_{\nu^{(h)}_R \wedge n}} \leq h\,,
\end{equation*}
for any $n$, where $\nu^{(h)}_R \wedge n := \min\{\nu^{(h)}_R,n\}$. Define \break $U_n := \sum_{i=1}^{n} (-g(R,x'_i,y'_i))$. From Wald's equation, Eq. \eqref{eq:expect}, and the fact that $g(r,x',y')$ is a nonincreasing function over $r \geq R$, we have
\begin{align*}
\frac{R}{4} \ee{\nu^{(h)}_R \wedge n ~\Big|~ h} &\leq \ee{-g(R,x'_i,y'_i)}\ee{\nu^{(h)}_R \wedge n ~\Big|~ h} \\
&= \ee{U_{\nu^{(h)}_R \wedge n} ~\Big|~ h} \\
&\leq \ee{\sum_{i = 1}^{\nu^{(h)}_R \wedge n} (-g(r_{i-1},x'_i,y'_i)) ~\Big|~ h} \leq h\,,
\end{align*}
for all $n$. As shown in the previous section, it follows that $\eee{\nu^{(h)}_R \mid h} = \lim_{n\to\infty} \eee{\nu^{(h)}_R \wedge n \mid h} \leq \frac{4h}{R}<\infty$. Similarly,
{\small
\begin{align*}
&\left(\ee{-g(R,x'_i,y'_i)}\right)^2\vv{\nu^{(h)}_R \wedge n ~\Big|~ h} \leq 2\bigg[ \vv{U_{\nu^{(h)}_R \wedge n} ~\Big|~ h} \\
&~ + \vv{(\nu^{(h)}_R \wedge n)\ee{-g(R,x'_i,y'_i)} - U_{\nu^{(h)}_R \wedge n} ~\Big|~ h}\bigg] \\
&\leq 2\bigg[\vv{U_{\nu^{(h)}_R \wedge n} ~\Big|~ h} + \ee{\nu^{(h)}_R \wedge n ~\Big|~ h}\vv{-g(R,x'_i,y'_i)} \bigg] \\
&\leq 2\left[h^2 + \frac{4h}{R}\frac{R^2}{4}\right]\,,
\end{align*}}
for all $n$, where the second inequality is due to Wald's identity (\cite{Resnik}, page $398$). Thus,
\begin{align}
\label{eq:finite variance}
\vv{\nu^{(h)}_R ~\Big|~ h} &= \lim_{n\to\infty} \vv{\nu^{(h)}_R \wedge n ~\Big|~ h} \nonumber \\
&\leq \frac{32h(h+R)}{R^2}<\infty\,.
\end{align}

This proves that the variance of path length generated by the random $\frac{1}{2}$disk routing scheme is finite almost surely. Next we will find some asymptotically tight bounds on the variance of the routing path lengths.

Let $S_{\nu} := \sum_{n=1}^{\nu} x'_n$ for a stopping time $\nu$ such that $\{\nu \geq N\}$ and $\{x'_n\}_{n < N}$ are independent and $\ee{\nu}<\infty$. Then by Wald's identity (\cite{Resnik}, page $398$) we have $\ee{S_{\nu}} = \ee{x'_n}\ee{\nu}$ and
\begin{align*}
\vv{\nu\ee{x'_n}-S_{\nu}} = \ee{(S_{\nu}-\nu\ee{x'_n})^2} = \ee{\nu}\vv{x'_n}\,.
\end{align*}

As such, we have
\begin{align*}
&\bigg| \sqrt{\vv{\nu}}\ee{x'_n} - \sqrt{\ee{\nu}\vv{x'_n}}\bigg| = \\
&\quad \left| \sqrt{\vv{\nu\ee{x'_n}}} - \sqrt{\vv{\nu\ee{x'_n}-S_{\nu}}}\right| \leq \sqrt{\vv{S_{\nu}}}.
\end{align*}

In particular, for $\nu = \nu^{(h)}_R$, we have
%{\small
\begin{equation}
\label{eq:crude bound on var(R)}
\left| \sqrt{\frac{\vv{\nu^{(h)}_R\mid h}}{\ee{\nu^{(h)}_R \mid h}}} - \sqrt{\frac{\vv{x'_n}}{(\ee{x'_n})^2}} \,\right| \leq \sqrt{\frac{\vv{S_{\nu^{(h)}_R}\mid h}}{\ee{\nu^{(h)}_R\mid h}(\ee{x'_n})^2}}\,.
\end{equation}%}

Hence, in order to prove the limit in \eqref{eq:bound on Var(R)}, we need to show that
\begin{equation*}
\frac{\vv{S_{\nu^{(h)}_R}\mid h}}{\ee{\nu^{(h)}_R\mid h}(\ee{x'_n})^2} \sim \frac{\vv{S_{\nu^{(h)}_R}\mid h}}{\frac{3\pi}{16}Rh} \to 0\,,
\end{equation*}
as $N\to\infty$. Suppose $R \leq a < b \leq h$ and note that
\begin{equation*}
-g(r_{n-1},x'_n,y'_n) \leq x'_n \leq -g(r_{n-1},x'_n,y'_n) + \frac{R^2}{2r_{n-1}}\,;
\end{equation*}
then together with \eqref{eq:ba}, we obtain
\begin{align*}
b-a &\leq \sum_{n = 1+\nu^{(a)}_R}^{\nu^{(b)}_R} (-g(r_{n-1},x'_n,y'_n)) \leq \sum_{n = 1+\nu^{(a)}_R}^{\nu^{(b)}_R} x'_n \\
&= S_{\nu^{(b)}_R} - S_{\nu^{(a)}_R} \\
&\leq \sum_{n = 1+\nu^{(a)}_R}^{\nu^{(b)}_R} (-g(r_{n-1},x'_n,y'_n)) + \sum_{n = 1+\nu^{(a)}_R}^{\nu^{(b)}_R} \frac{R^2}{2r_{n-1}} \\
&\leq b - a + R + \frac{R^2}{2a}\nu^{(b)}_R\,,
\end{align*}
where the last inequality is due to the fact that $r_n \geq a$ for $\nu^{(a)}_R \leq n \leq \nu^{(b)}_R$. Therefore, we obtain
\begin{align*}
&\sqrt{\vv{S_{\nu^{(b)}_R} - S_{\nu^{(a)}_R}~\Big|~ a,b} } \\
&\qquad= \sqrt{\ee{\left[S_{\nu^{(b)}_R} - S_{\nu^{(a)}_R}- \ee{S_{\nu^{(b)}_R} - S_{\nu^{(a)}_R}}\right]^2 ~\Big|~ a,b}} \\
&\qquad\leq \sqrt{\ee{\left[S_{\nu^{(b)}_R} - S_{\nu^{(a)}_R} - b + a \right]^2 ~\Big|~ a,b}} \\
&\qquad\leq \sqrt{\ee{\left[R + \frac{R^2}{2a}\nu^{(b)}_R\right]^2 ~\Big|~ a,b}} \\
&\qquad\leq R + \frac{R^2}{2a}\ee{\nu^{(b)}_R\mid b} + \frac{R^2}{2a}\sqrt{\vv{\nu^{(b)}_R \mid b}} \\
%&\qquad\leq R + \frac{2Rb}{a} + \frac{R}{2a}\sqrt{32b(b+R)} \\
&\qquad\leq 6R + \frac{5Rb}{a}\,,
\end{align*}
using \eqref{eq:finite variance} and the fact that $\eee{\nu^{(b)}_R\mid b} \leq \frac{4b}{R}$ and $\vvv{\nu^{(b)}_R \mid b} \leq \frac{32b(b+R)}{R^2}$. Finally, we let $a_i = R\left(\frac{h}{R}\right)^{i/k}$, for $k = \lceil \log \frac{h}{R} \rceil$ and $i = 0,1,2,\ldots,k$, where $\lceil \log \frac{h}{R} \rceil$ is the smallest integer larger than $\log \frac{h}{R}$. Then we have

\begin{align}
\label{eq:var_offset}
\sqrt{\vv{S_{\nu^{(h)}_R}\mid h}} &\leq \sum_{i=1}^{k} \sqrt{\vv{S_{\nu^{(a_i)}_R} - S_{\nu^{(a_{i-1})}_R}~\Big|~ h}} \nonumber \\
&\leq 6kR + 5R\sum_{i=1}^{k} \frac{a_i}{a_{i-1}} \nonumber \\
%&= 6kR + 5kR\left(\frac{h}{R}\right)^{1/k} \nonumber \\
&\leq (6+5e)(1+\log \frac{h}{R})R.
\end{align}

From this, it follows that
\begin{equation*}
\sqrt{\frac{\vv{S_{\nu^{(h)}_R}\mid h}}{Rh}} \leq (6+5e)(1+\log \frac{h}{R})\sqrt{\frac{R}{h}} \to 0\,
\end{equation*}
as $N\to\infty$, which concludes the proof for the limit in Eq. \eqref{eq:bound on Var(R)}.

\begin{remark}
\label{rem:path-stretch}
It is worth noting that the path-stretch statistics can be easily derived from the hop-count statistics: Let $\mathcal{L}_{\nu_R^{(h)}}$ denote the path-stretch of a routing path with length $\nu_R^{(h)}$ connecting a source-destination pair that is $h$-distance apart, i.e., $\mathcal{L}_{\nu_R^{(h)}} := \|X_1-X_0\|+\|X_2-X_1\|+\cdots+\|X_{\nu_R^{(h)}+1}-X_{\nu_R^{(h)}}\|$. Then, it is easy to show that $\eee{\mathcal{L}_{\nu_R^{(h)}}\mid h} \sim  \frac{\pi}{2}h$ and $\vvv{\mathcal{L}_{\nu_R^{(h)}}\mid h} \sim \frac{\pi}{12}Rh$ as $N\to\infty$. Therefore, in the case of a dense network, $\mathcal{L}_{\nu_R^{(h)}}\to \frac{\pi}{2}h$ a.a.s. since $\vvv{\mathcal{L}_{\nu_R^{(h)}}}\to 0$ as $N\to\infty$.
\end{remark}

\section{Simulation Results}
\label{sec:Simulation Results}

In this section we compare our analytical results with some empirical results derived through simulation. In Fig. \ref{fig:routingsamples}, we depict some realizations for the routing paths generated by the random $\frac{1}{2}$disk routing scheme for an arbitrary source located at $(-1/4,-1/4)$ and its destination at $(1/4,1/4)$ with the following network specifications: $|A| = 1$, $\lambda = 10^{6}$, and $R = \sqrt{\frac{2\log \lambda}{\lambda}} \approx 5.2\times 10^{-3}$. As illustrated in this figure, the path realizations do not closely follow the direct line connecting the source-destination nodes. The lengthes of the routing paths are $328$, $314$, $343$ for the realizations depicted in Fig. \ref{fig:routingsamples} (a), (b), and (c) respectively. Fig. \ref{fig:routingsamples} (d) depicts an ensemble of thirty realizations of the random $\frac{1}{2}$disk routing scheme. Based on \eqref{eq:E(nu)-bound} we obtain the lower and upper bounds of $314$, $370$ for the expected path length with the asymptotic value of $317$. (Note that the bounds derived in \eqref{eq:E(nu)-bound} are for the expected path length; therefore, individual realizations for the path length might violate these bounds.)
\begin{figure}[ht]
  \centering
  \includegraphics[width=3.2in]{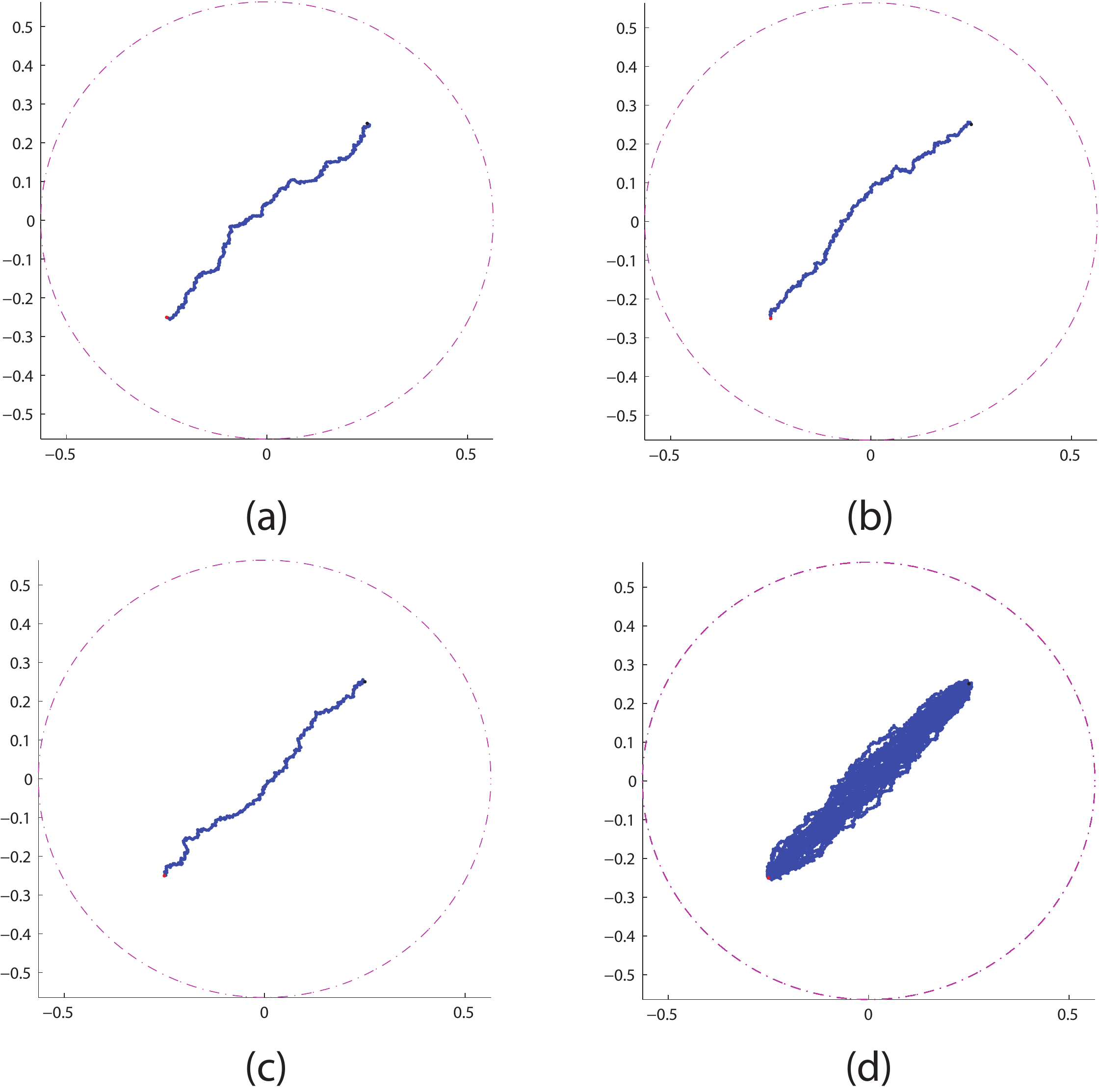}
  \caption{Random $\frac{1}{2}$disk routing realizations for $\lambda = 10^{6}$, $|A| = 1$, and $R = \sqrt{\frac{2\log \lambda}{\lambda}}$, when the source is located at $(-1/4,-1/4)$ and its destination is located at $(1/4,1/4)$. The dashed circle demonstrates the network boundary.}
  \label{fig:routingsamples}
\end{figure}

The following empirical path length statistics are obtained by generating $100$ realizations of the network via placing $N \sim \textrm{Poisson}(\lambda)$ nodes uniformly over a circular disk with unit area. For each network realization we constructed $100$ realizations for the random $\frac{1}{2}$disk routing path: starting from a fixed source node, we find the subsequent relaying nodes according to the rRSL scheme until (possibly) reaching the fixed destination node. Source and destination are set $h = \sqrt{2}/2$ distance apart and the transmission ranges are chosen as $R = \sqrt{\frac{2\log \lambda}{\lambda}}$. In Fig. \ref{fig:ExpectedLength}, we compare the (normalized) empirical mean, $\frac{R}{h}\eee{\nu^{(h)}_R}$, of the path lengths generated by the random $\frac{1}{2}$disk routing scheme with the analytical bounds derived in Eq. \eqref{eq:E(nu)-bound} for different values of network node density. As shown in this figure, the normalized empirical mean converges to $3\pi/4 \approx 2.3562$, and is always bounded by the expressions derived in Eq. \eqref{eq:E(nu)-bound}.
\begin{figure}[ht]
  \centering
  \includegraphics[width=3.2in]{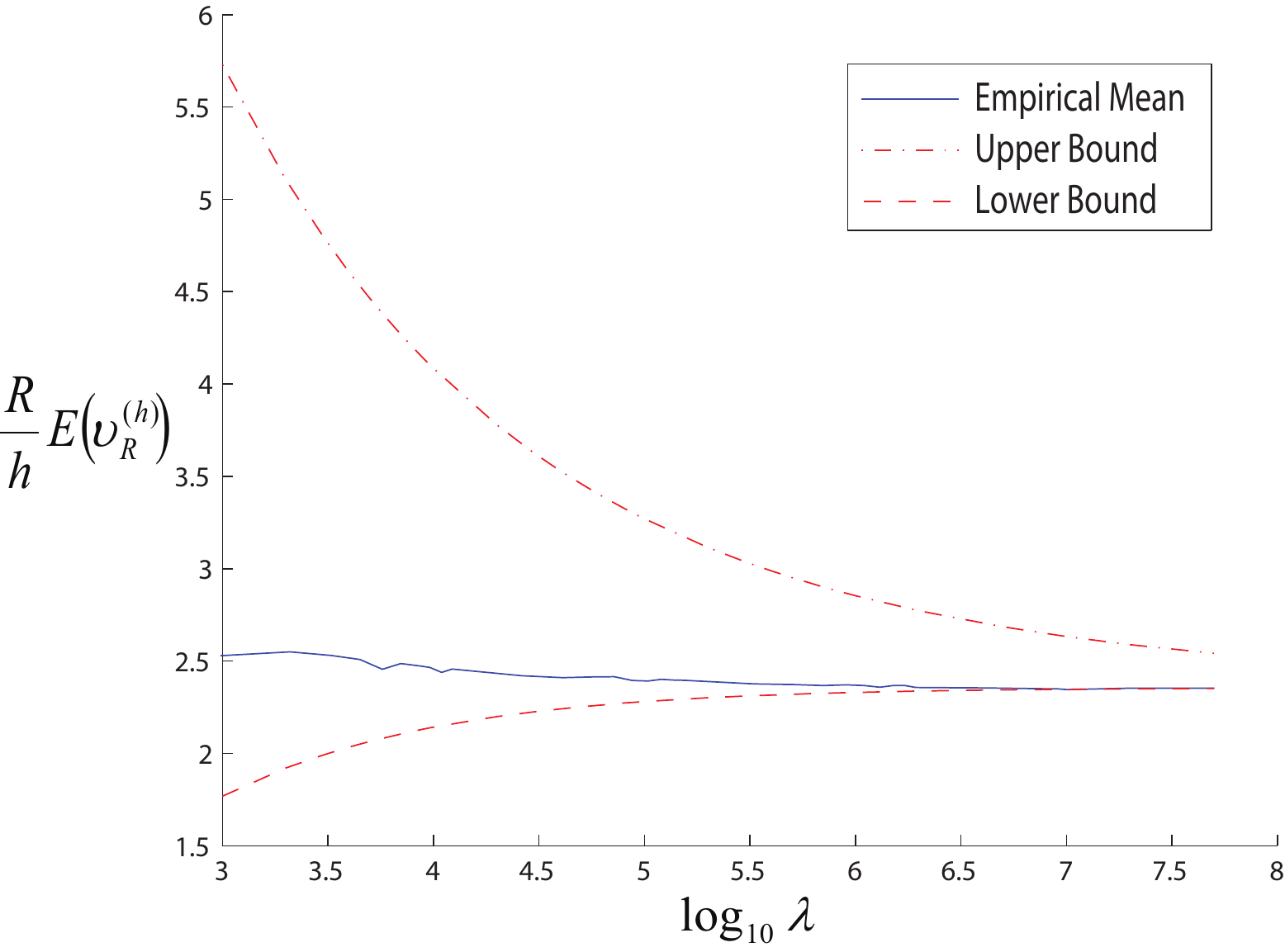}
  \caption{Numerical comparison between analytical bounds derived in Eq. \eqref{eq:E(nu)-bound} and the (normalized) empirical mean of the path length generated by the random $\frac{1}{2}$disk routing scheme when $h = \sqrt{2}/2$, $|A| = 1$, and $R = \sqrt{\frac{2\log \lambda}{\lambda}}$.}
  \label{fig:ExpectedLength}
\end{figure}

In Fig. \ref{fig:Lengthvariance}, we compare the empirical variance-to-mean ratio of the random $\frac{1}{2}$disk routing scheme, $\sqrt{\vvv{\nu^{(h)}_R}/\eee{\nu^{(h)}_R}}$, with the analytical bounds derived in Eq. \eqref{eq:var_offset} for different values of network node density. As shown in this figure, the normalized empirical standard deviation converges to $\sqrt{9\pi^2/64 - 1} \approx 0.6228$, and is always bounded by the expressions derived in Eq. \eqref{eq:var_offset}. Furthermore, it can be seen that although the bounds in \eqref{eq:var_offset} are quite loose for small values of $\lambda$, the asymptotic standard deviation derived in \eqref{eq:bound on Var(R)} is very close to the empirical standard deviation even for small values of $\lambda$.
\begin{figure}[ht]
  \centering
  \includegraphics[width=3.2in]{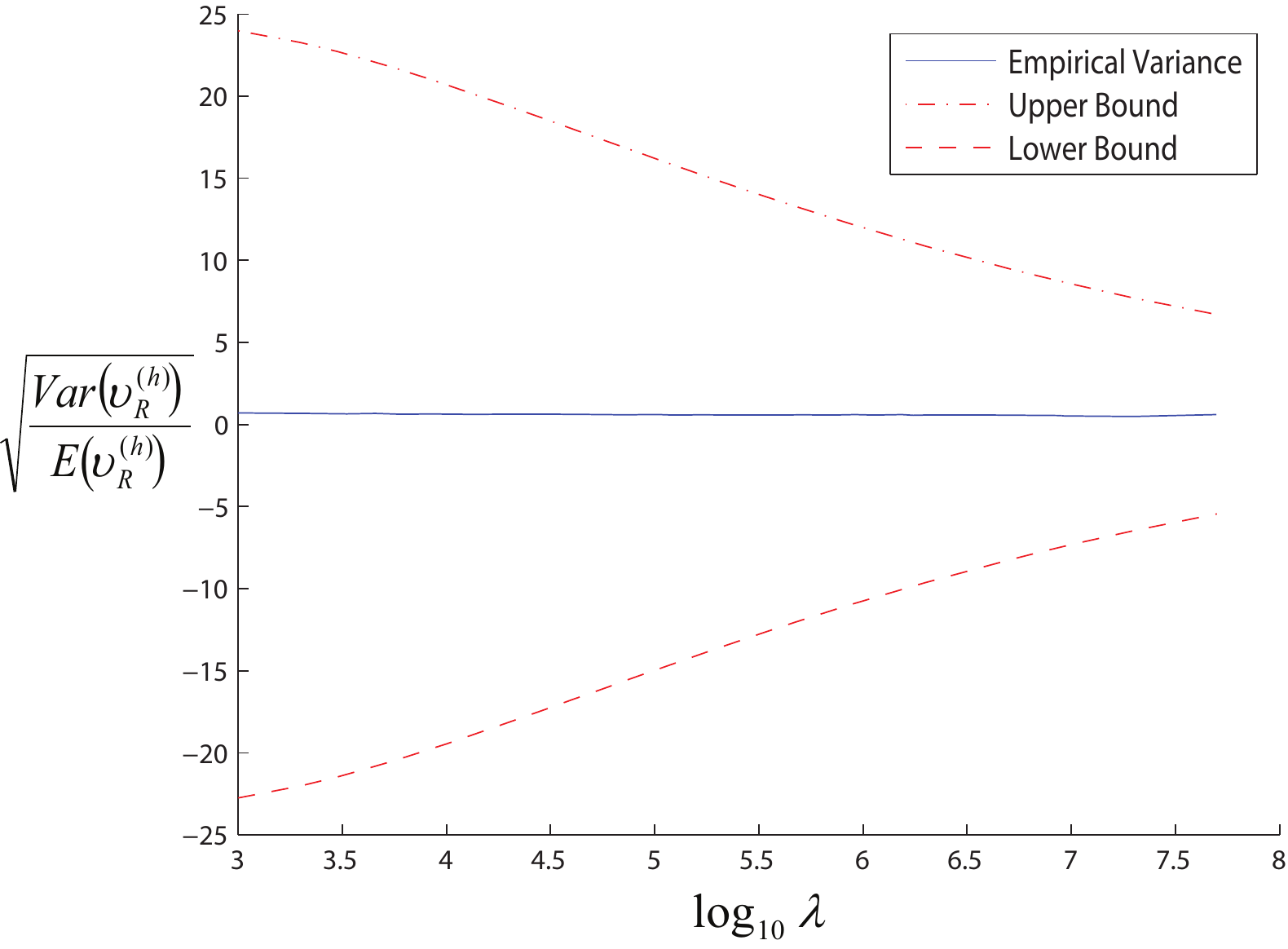}
  \caption{Numerical comparison between analytical bounds derived in Eq. \eqref{eq:crude bound on var(R)} and the (normalized) empirical standard deviation of the path length generated by the random $\frac{1}{2}$disk routing scheme, when $h = \sqrt{2}/2$, $|A| = 1$, and $R = \sqrt{\frac{2\log \lambda}{\lambda}}$.}
  \label{fig:Lengthvariance}
\end{figure}

In Fig. \ref{fig:Length}, we demonstrate the deviation of the path length realizations from its asymptotic expected value for different values of network node density. As shown in this figure, the deviation of the path length realizations increases as the network density and consequently the expected length of the routing path increases. However, all realizations stay relatively close to the value predicted by Eq. \eqref{eq:expected length_2}.
\begin{figure}[ht]
  \centering
  \includegraphics[width=3.2in]{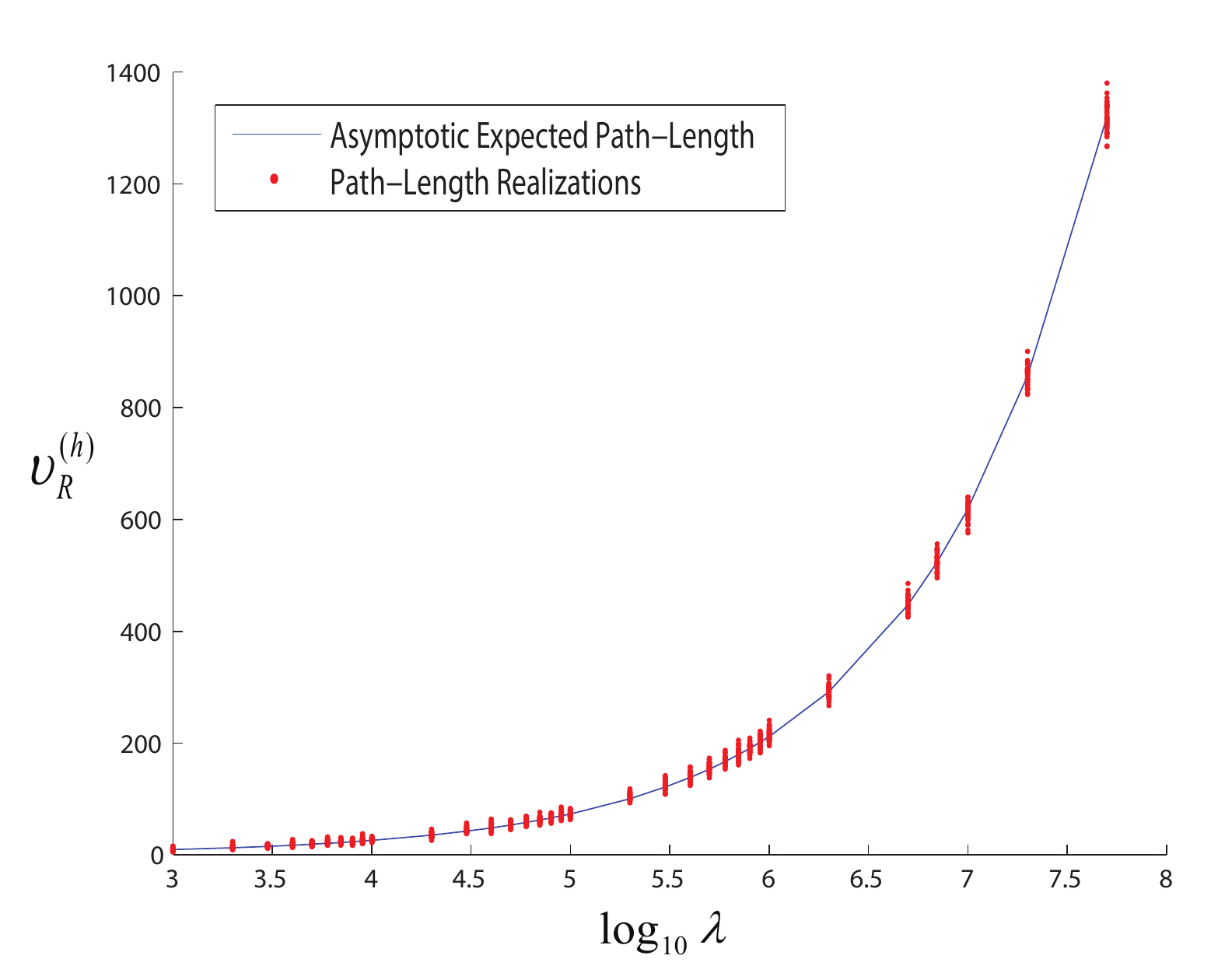}
  \caption{Random $\frac{1}{2}$disk routing realizations for $\lambda = 10^{6}$, $|A| = 1$, and $R = \sqrt{\frac{2\log \lambda}{\lambda}}$, when the source is located at $(-1/4,-1/4)$ and its destination is located at $(1/4,1/4)$.}
  \label{fig:Length}
\end{figure}

As mentioned earlier, we ignored the edge effect when computing the asymptotic path length statistics of the random $\frac{1}{2}$disk routing scheme. In Figs. \ref{fig:meanHop_edge} and \ref{fig:varHop_edge}, we consider two source-destination pairs that are close to the network edge with different distances and investigate whether routing ``next to the boundary'' has a considerable impact on the length of the routing paths. We consider a source node $S$ at $(-0.379,-0.379)$ and two destination nodes $Dst_1$ at $(0.3267,-0.4234)$ and $Dst_2$ at $(-0.315,-0.4336)$ such that $h_1 = \|S-Dst_1\| = \sqrt{2}/2$ and $h_2 = \|S-Dst_2\| = \sqrt{2}/33$. Note that $\|S\| = 0.95L$, $\|Dst_1\| = 0.948L$, and $\|Dst_2\| = 0.95L$. Fig. \ref{fig:meanHop_edge} depicts the empirical mean and Fig. \ref{fig:varHop_edge} depicts the empirical variance-to-mean ratio of paths generated by the random $\frac{1}{2}$disk routing scheme for source-destination pairs a) $S-Dst_1$ and b) $S-Dst_2$. Comparing these figures with Figs. \ref{fig:ExpectedLength} and \ref{fig:Lengthvariance}, we observe that given a fixed $h$, routing close to the network edge does not affect the \emph{asymptotic} path statistics. Intuitively, as shown in Remark \ref{rem:h}, the distances between source-destination pairs will be of order $h = \Theta(L)$ with high probability where $h/R\to 0$ as $N\to\infty$. Therefore, for large enough $N$, it is very unlikely that a considerable portion of the path connecting a source to its destination traverses close to the network edge. As such, the effect of the routing close to the boundary on path statistics is relatively negligible for large network sizes. However, for small network sizes (when $h$ and $R$ are comparable), the empirical mean of the path length is smaller than the value predicted in \eqref{eq:expected length_1}.
\begin{figure}[ht]
  \centering
  \includegraphics[width=3.5in]{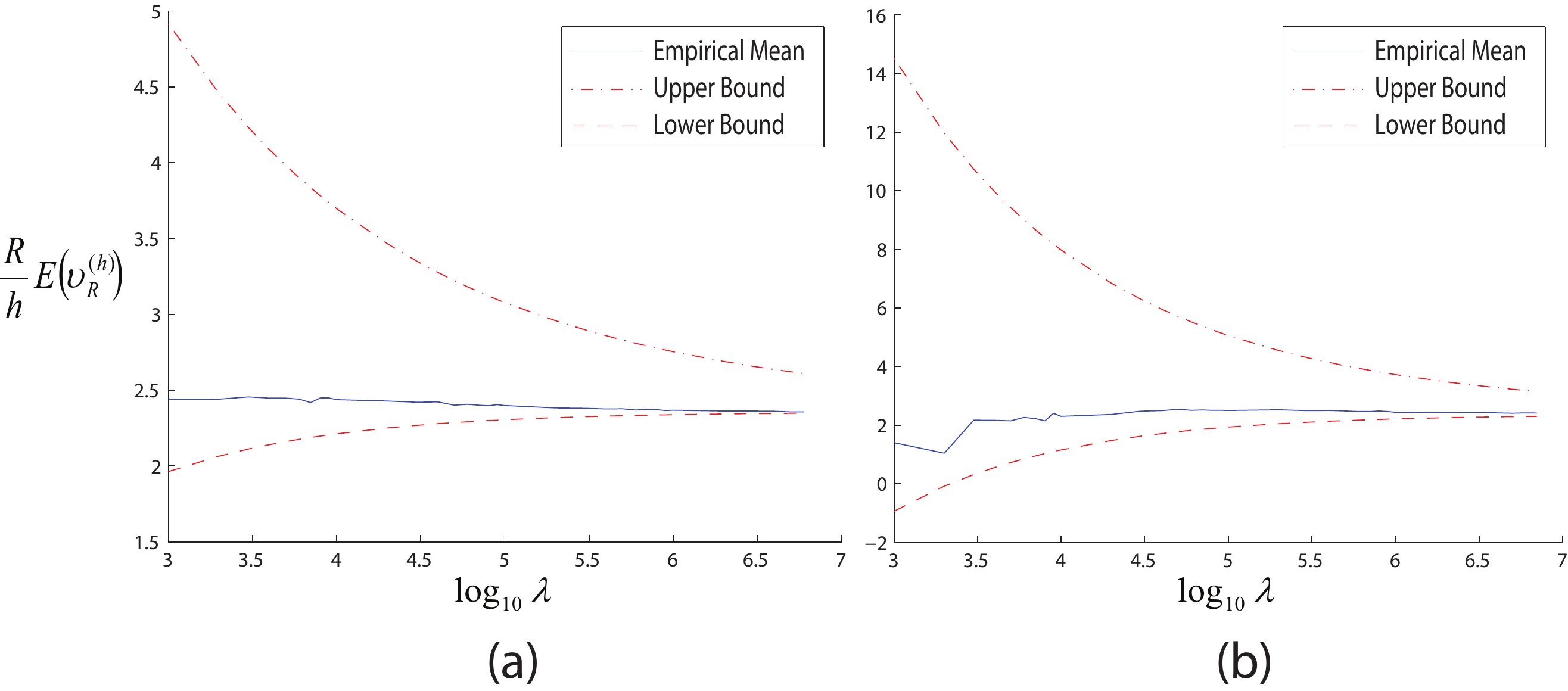}
  \caption{Numerical comparison between analytical bounds derived in Eq. \eqref{eq:E(nu)-bound} and the (normalized) empirical mean of the path length generated by the random $\frac{1}{2}$disk routing scheme for source-destination pairs that are close to the network boundary when a) $h = \sqrt{2}/2$ and b) $h = \sqrt{2}/33$. In both cases $|A| = 1$, $R = \sqrt{\frac{2\log \lambda}{\lambda}}$, and $\|S\| = \|Dst\| \simeq 0.95L$.}
  \label{fig:meanHop_edge}
\end{figure}
\begin{figure}[ht]
  \centering
  \includegraphics[width=3.5in]{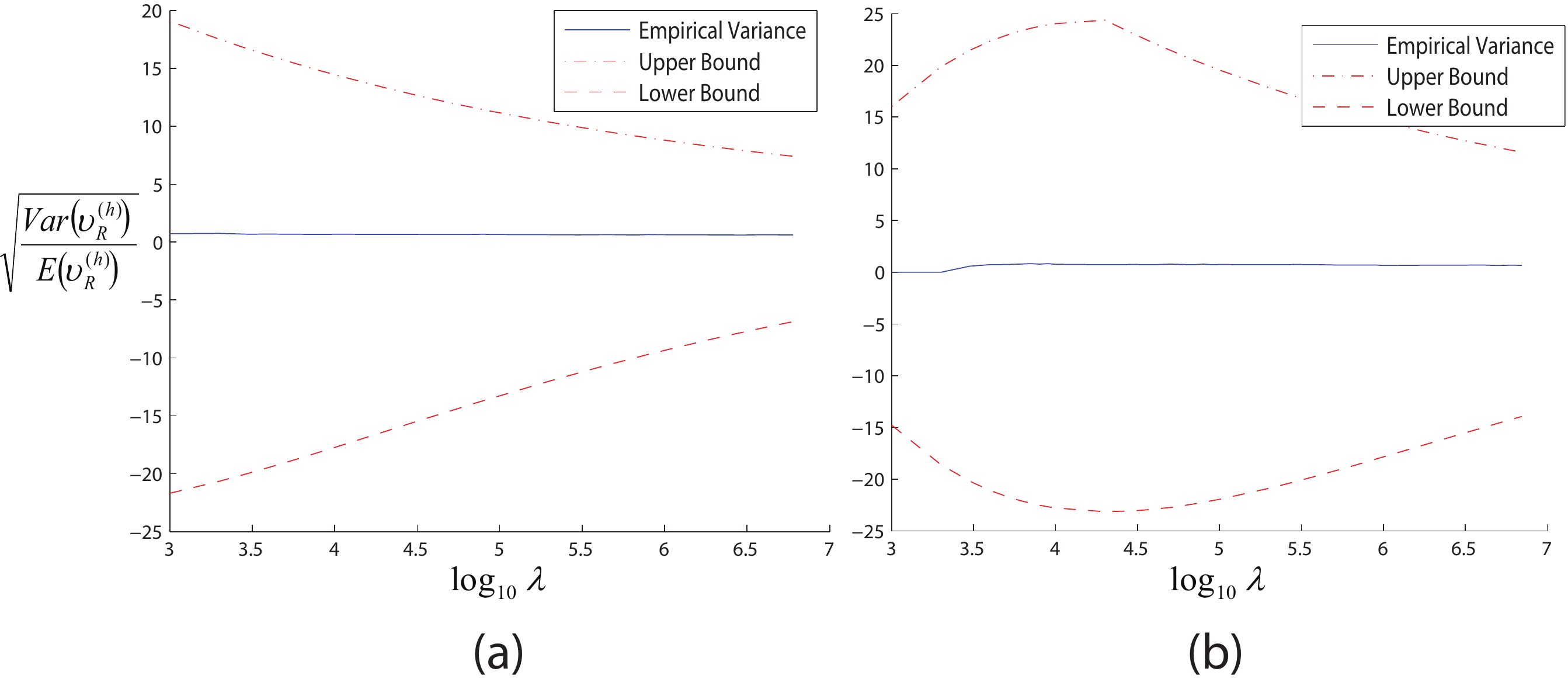}
  \caption{Numerical comparison between analytical bounds derived in Eq. \eqref{eq:crude bound on var(R)} and the (normalized) empirical standard deviation of the path length generated by the random $\frac{1}{2}$disk routing scheme for source-destination pairs that are close to the network boundary when a) $h = \sqrt{2}/2$ and b) $h = \sqrt{2}/33$. In both cases $|A| = 1$, $R = \sqrt{\frac{2\log \lambda}{\lambda}}$, and $\|S\| = \|Dst\| \simeq 0.95L$.}
  \label{fig:varHop_edge}
\end{figure}

\section{Generalization}
\label{sec:generalization}

In the previous sections we derived sufficient conditions for the network to be connected deploying the random $\frac{1}{2}$disk routing scheme and quantified the mean and variance asymptotes of the routing path generated the random $\frac{1}{2}$disk routing scheme. In this section we present some guidelines that generalize the aforementioned results for some other variants of the geometric routing schemes such as MFR, DIR, NFP, and the random $\eta$disk routing scheme, where the latter one is the generalized version of the random $\frac{1}{2}$disk routing scheme with an $\eta$disk as its RSR.

Observe that the results of Section \ref{sec:connectivity} were derived for the general $\eta$disks relay selection region which encompasses most of the geometric routing schemes such as MFR, DIR, NFP, and the random $\eta$disk routing scheme. Let $\Delta$ be the set of all nodes (in the RSR of a specific transmitting node) that can be selected as the next relay by the relay selection rule (RSL) of the geometric routing scheme. For example, in the cases of MFR, DIR, NFP, and the random $\eta$disk routing scheme we have: $\Delta_{\textrm{MFR}} := \{(x'_n,y'_n)\in \frac{1}{2}\textrm{RSR}:\, x'_n \geq x, \textrm{ for all } (x,y)\in \frac{1}{2}\textrm{RSR}\}$, $\Delta_{\textrm{DIR}} := \{(x'_n,y'_n)\in \frac{1}{2}\textrm{RSR}:\, |\tan^{-1}(y'_n/x'_n)|\leq |\tan^{-1}(y/x)|, \textrm{ for all } (x,y)\in \frac{1}{2}\textrm{RSR}\}$, $\Delta_{\textrm{NFP}} := \{(x'_n,y'_n)\in \frac{1}{2}\textrm{RSR}:\, \sqrt{(x'_n)^2+(y'_n)^2} \leq \sqrt{(x)^2+(y)^2}, \textrm{ for all } (x,y)\in \frac{1}{2}\textrm{RSR}\}$, and $\Delta_{\eta} = \{(x'_n,y'_n)\in \eta\textrm{RSR}\}$, respectively. Since the nodes in $\Delta$ (if more than one) are indistinguishable by the RSL, the transmitting node selects one of the nodes in $\Delta$ randomly as the next relay. Next, we present the generalized results for the network connectivity and the mean and variance asymptotes of routing paths generated by the general geometric routing schemes.

\begin{coro}
\label{cor: Generalconnectivity}
Let $\Delta$ be the set of all nodes that can be selected by the relay selection rule as the next relay. Then the network is connected employing the geometric routing scheme a.a.s. if $\ee{g(R,x',y')\I{\Delta}} < 0$.
\end{coro}
\begin{proof}
The proof is immediate due to \eqref{eq:genralbound}.
\end{proof}

\begin{coro}
\label{cor: Generalmean}
If $\ee{g(R,x',y')\I{\Delta}} < 0$ and $\ee{(y')^2\I{\Delta}} \leq R\,\ee{x'\I{\Delta}}$, the expected length of the routing path generated by the general geometric routing scheme connecting a source-destination pair that is $h$-distance apart scales as $\ee{\nu\mid h} \sim h / \ee{x'\I{\Delta}}$ as $N\to\infty$.
\end{coro}
\begin{proof}
The proof follows directly from \eqref{eq:bound on g} and noting that if $\ee{(y')^2\I{\Delta}} \leq R\,\ee{x'\I{\Delta}}$, using the intermediate value theorem, we can find $r$ such that $\frac{2R(r-h)}{\ee{(y')^2\I{\Delta}}} = \frac{R}{\ee{x'\I{\Delta}}}\,(\sqrt{\frac{h}{2R}}+1)$, which yields the bound in Eq. \eqref{eq:E(nu)-bound} and hence the desired result.
\end{proof}

\begin{coro}
\label{cor: Generalvar}
If $\ee{g(R,x',y')\I{\Delta}} < 0$, the variance of the path length generated by the general geometric routing scheme, normalized by its mean, scales as $\vv{\nu}/\ee{\nu} \sim \vv{x'\I{\Delta}}/\left(\ee{x'\I{\Delta}}\right)^2$ as $N\to\infty$.
\end{coro}
\begin{proof}
The proof follows the same steps as in Section \ref{subsec:Length Variance}.
\end{proof}

\section{Conclusion}
\label{sec:conclusion}

In this paper, we presented a simple methodology employing statistical analysis and stochastic geometry to study geometric routing schemes in wireless ad-hoc networks, and in particular, analyzed the network layer performance of one such scheme named the random $\frac{1}{2}$disk routing scheme. We defined a notion of network connectivity considering the special local properties of geometric routing schemes and determined some sufficient conditions that guarantee network connectivity when each node finds its next relay in the so-defined $\frac{1}{2}$disk. More specifically, if all nodes transmit at a power that covers a normalized area $d$ and the expected number of nodes in the network is $N$, the network is connected a.a.s. if $\eta dN + \log d \to \infty$ when $N\to\infty$. Furthermore, we proved that the routing path progress conditioned on the previous two hops can be approximated with a Markov process. Then using this Markovian approximation, we derived exact asymptotic expressions for the mean and variance of the path length generated by the random $\frac{1}{2}$disk routing scheme. Furthermore, we provided guidelines to extend these results to other variants of geometric routing schemes such as MFR, DIR, and NFP.

\appendices
\section{Proof of Proposition \ref{prop1}}
\label{appendix:markov}

First, let us consider the distribution of a Poisson point process conditioned on deleting one point. Let $\Phi$ be a homogeneous Poisson point process with intensity $\lambda$ and assume a fixed region $D$. If $\phi(D) > 0$, one point in $D$ is selected at random and removed. Let $X$ be the location of that point. The distribution of $\Phi$ on $D^c$ remains Poisson and independent of $\Phi$ on $D$, and thus independent of both $\phi(D)$ and $X$. Let $\Phi'$ be the (point) process with the point at $X$ deleted. (Note that the distribution of $\Phi'$ is not the same as the reduced Palm distribution \cite{Baccelli2} of $\Phi$, as the location of node $X$ is random.)

Let $A_1, A_2, \ldots, A_k$ be a partition of $D$. Given $\phi(D) > 0$, the points in $D$ are distributed uniformly. If one point is removed at random, the remaining points are still distributed uniformly on $D$. Hence,
\begin{align}
\label{eq:dist1}
&\p{\bigcap_{j=1}^{k}\{\phi'(A_j) = n_j\}~ \Big|~ \phi(D) > 0, X}= \nonumber \\
&~ (1 - e^{-\lambda |D|})^{-1} \sum_{i=1}^{k} \frac{n_i+1}{n_1+\ldots+n_k+1}\, \cdot \nonumber \\
&~ \prod_{j=1}^{k} \frac{(\lambda |A_j|)^{n_j+\I{j=i}}}{(n_j+\I{j=i})!}e^{-\lambda |A_j|} = \nonumber \\
&~ \frac{\lambda |D|}{(1 - e^{-\lambda |D|})(n_1+\ldots+n_k+1)} \prod_{j=1}^{k}\frac{(\lambda |A_j|)^{n_j}}{(n_j)!}e^{-\lambda |A_j|}\,,
\end{align}
since $|A_1|+\ldots+ |A_k| = |D|$. Therefore, conditional on $\phi(D)>0$, $\Phi'$ is independent of the location of the removed point ($X$). In particular,
\begin{align*}
&\p{\phi'(D) = n ~\Big|~ \phi(D)>0,X}= \\
&\frac{(\lambda |D|)^{n+1}}{(n+1)!(1 - e^{-\lambda |D|})}e^{-\lambda |D|} = \\
&\p{\phi(D) = n+1 ~\Big|~ \phi(D)>0}\,.
\end{align*}
Furthermore, given $n_1 + \ldots + n_k = n > 0$, we have
\begin{align*}
&\p{\bigcap_{j=1}^{k}\{\phi'(A_j) = n_j\}~ \Big|~ \phi(D) > 0, \phi'(D) = n, X}= \\
&{n\choose n_1\cdots n_k} \prod_{j=1}^{k}\left(\frac{|A_j|}{|D|}\right)^{n_j}\,.
\end{align*}
Thus, for $A \subseteq D$ and given $\phi'(D) = n > 0$, $\phi'(A)$ is conditionally $\textrm{Binomial}\left(n,\frac{|A|}{|D|}\right)$. Without knowing $\phi'(D)$, however, we obtain from \eqref{eq:dist1} that
\begin{align}
\label{eq:dist2}
&\p{\phi'(A) = k ~\Big|~ \phi(D)>0,X}= \nonumber \\
&\frac{\lambda |D|e^{-\lambda |D|}}{(1 - e^{-\lambda |D|})} \sum_{j=0}^{\infty} \frac{1}{k+j+1}\frac{(\lambda |A|)^k}{k!}\frac{(\lambda |A^c\cap D|)^j}{j!}= \nonumber \\
&\frac{\lambda |D|e^{-\lambda |D|}}{(1 - e^{-\lambda |D|})} \frac{(\lambda |A|)^k}{k!}\int_{0}^{1}y^ke^{\lambda |A^c\cap D|y}\D{y}\,,
\end{align}
where the second equality is due to
\begin{align*}
\sum_{j=0}^{\infty} \frac{1}{k+j+1}\frac{a^j}{j!} &= \sum_{j=0}^{\infty} \frac{1}{a^{k+1}}\int_{0}^{a} \frac{x^{k+j}}{j!}\D{x} \\
&= \int_{0}^{a} \frac{x^k}{a^{k+1}}e^x\D{x} = \int_0^1 y^ke^{ay}\D{y}\,.
\end{align*}

\begin{figure*}[ht]
  \begin{subequations}
  \begin{align}
  \hline \nonumber \\
  \ee{\frac{\phi'(CD)}{\phi'(C)}\I{\phi'(C)>0}~\big|~\phi(D)>0, X}
  &= \sum_{n=0}^{\infty}\sum_{m=1}^{\infty} \frac{n}{n+m} \frac{(\lambda |CD^c|)^m}{m!}e^{-\lambda|CD^c|} \frac{\lambda |D| e^{-\lambda|D|}}{1-e^{-\lambda|D|}} \frac{(\lambda |CD|)^n}{n!} \int_0^1 y^n e^{\lambda |C^cD|y}\D{y} \nonumber \\
  %&= \frac{\lambda |D| e^{-\lambda|C\cup D|}}{1-e^{-\lambda|D|}} \sum_{n=1}^{\infty} \frac{(\lambda |CD|)^n}{(n-1)!} \int_0^1 w^{n-1}\left(e^{\lambda |CD^c|w}-1\right) \D{w} \int_0^1 y^n e^{\lambda |C^cD|y}\D{y}  \nonumber \\
  &= \frac{\lambda |D| e^{-\lambda|C\cup D|}}{1-e^{-\lambda|D|}} \int_0^1 \int_0^1 \lambda |CD|ye^{\lambda\left(|CD|yw+|C^cD|y\right)}\left(e^{\lambda|CD^c|w} -1\right)\D{w}\D{y} \label{eq:inter_a}\\
  %&< \frac{\lambda |D| e^{-\lambda|C\cup D|}}{1-e^{-\lambda|D|}} \int_0^1 \int_0^1 \lambda |CD|ye^{\lambda\left(|CD|yw+|C^cD|y+|CD^c|w\right)}\D{w}\D{y} \nonumber \\
  &< \frac{\lambda |D| e^{-\lambda|C\cup D|}}{1-e^{-\lambda|D|}} \int_0^1 \frac{|CD|y}{|CD|y+|CD^c|}\left(e^{\lambda(|CD|y+|CD^c|)}-1\right)e^{\lambda |C^cD|y}\D{y} \nonumber\\
  %&< \frac{|CD|}{|C|} \frac{\lambda |D| e^{-\lambda|C\cup D|}}{1-e^{-\lambda|D|}} \int_0^1 \left(e^{\lambda(|CD|y+|CD^c|)}-1\right)e^{\lambda |C^cD|y}\D{y} \nonumber \\
  &< \frac{|CD|}{|C|} \left(1 - \frac{|D|}{|C^cD|}\frac{1 - e^{-\lambda |C^cD|}}{1-e^{-\lambda|D|}}e^{-\lambda|C|}\right)\,. \label{eq:inter_b} %\nonumber \\
  %&= \frac{|CD|}{|C|} \p{\phi'(C) > 0 ~\Big|~ \phi(D)>0, X}\,. \label{eq:inter_b}
  %\\ \nonumber
  %\\ \hline \nonumber
  \end{align}
  \end{subequations}
  \begin{align}
  %\hline \nonumber \\
  \textrm{E}\bigg(\frac{\phi'(CD)}{\phi'(C)}\I{\phi'(C)>0}~\big|~ \phi(D)>0, X\bigg) %= \frac{\lambda |D| e^{-\lambda|C\cup D|}}{1-e^{-\lambda|D|}} \bigg[\int_0^1 \frac{|CD|y}{|CD|y+|CD^c|}\left(e^{\lambda(|CD|y+|CD^c|)}-1\right)e^{\lambda |C^cD|y}\D{y} \nonumber \\
  %&- \int_0^1 \int_0^1 \lambda |CD|ye^{\lambda\left(|CD|yw+|C^cD|y\right)}\D{w}\D{y} \bigg] \nonumber \\
  %&> \frac{\lambda |D| e^{-\lambda|C\cup D|}}{1-e^{-\lambda|D|}} \left[ \int_0^1 \frac{|CD|}{|C|} y\left(e^{\lambda(|CD|y+|CD^c|)}-1\right)e^{\lambda |C^cD|y}\D{y} - \int_0^1 e^{\lambda |C^cD|y} \left(e^{\lambda |CD|y} - 1\right)\D{y}\right]\nonumber \\
  &> \frac{|CD|}{|C|}\frac{1}{(1-e^{-\lambda|D|})} \bigg\{\lambda |D| e^{-\lambda |D|} \int_0^1 ye^{\lambda |D| y} \D{y} - \lambda |D| e^{-\lambda|C\cup D|}\int_0^1 ye^{\lambda |C^c D| y} \D{y} \nonumber \\
  &~+ \frac{|C|}{|CD|}\left(-e^{-\lambda |CD^c|}(1-e^{-\lambda |D|}) + \frac{|D|}{|C^cD|}\left(e^{-\lambda|C|}-e^{-\lambda |C\cup D|}\right)\right)\bigg\} \nonumber \\
  %&> \frac{|CD|}{|C|} \frac{1}{(1-e^{-\lambda|D|})} \bigg\{1 - \frac{1}{\lambda |D|} - \frac{|D|}{|C^c D|}e^{-\lambda |C|}+ \frac{|C|}{|CD|}\left(-e^{-\lambda |CD^c|}(1-e^{-\lambda |D|}) + \frac{|D|}{|C^cD|}\left(e^{-\lambda|C|}-e^{-\lambda |C\cup D|}\right)\right) \bigg\} \nonumber \\
  %&> \frac{|CD|}{|C|} \frac{1}{(1-e^{-\lambda|D|})} \left[1 - \frac{1}{\lambda |D|} - \frac{|C|}{|CD|}\left(e^{-\lambda |CD^c|} + \frac{|D|}{|C^cD|}e^{-\lambda |C\cup D|}\right)\right] \nonumber \\
  &> \frac{|CD|}{|C|} \frac{1}{(1-e^{-\lambda|D|})} \left(1 - \frac{1}{\lambda |D|} - \frac{|C||D|}{|CD||C^cD|}e^{-2\lambda |CD^c|}\right)\,. \label{eq:ineq2}
  \\ \nonumber
  \\ \hline \nonumber
  \end{align}
\end{figure*}

%\newpage
After the aforementioned preliminaries, we now proceed with the proof of Proposition \ref{prop1}. Suppose $C$ is a random set that depends only on $X$.\footnote{Note that $D$ and $C$ here correspond to $D_{n-1}$ and $D_n$ in Section \ref{sec:Path Length Statistics}, respectively.} The points of $\Phi'$, if any, which are in $CD := C\cap D$, are uniformly distributed and independent of the points in $CD^c$, which are also uniformly distributed (if any). The combined points are uniformly distributed on $C$ \emph{only if} the expected proportion of points in $CD$ is $\frac{|CD|}{|C|}$.

However, the expected proportion of points in $CD$ is strictly less than $\frac{|CD|}{|C|}$ in our case as we now compute. Given $\phi'(C) > 0$, the probability that a randomly selected point in $C$ is also in $D$ is $\eee{\frac{\phi'(CD)}{\phi'(C)}~\big|~\phi'(C)>0, \phi(D)>0, X}$. Let $\frac{\phi'(CD)}{\phi'(C)}=0$ when $\phi'(C)=0$. Using \eqref{eq:dist2}, we have
\begin{align*}
&\p{\phi'(C) > 0 ~\Big|~ \phi(D)>0, X}=  \\
& 1 -  \p{\phi'(CD)=0, \phi'(CD^c)= 0 ~\Big|~ \phi(D)>0, X}=  \\
& 1 - \frac{\lambda |D| e^{-\lambda|D|}}{1-e^{-\lambda|D|}}\frac{e^{\lambda|C^cD|}-1}{\lambda|C^cD|}e^{-\lambda|CD^c|}=  \\
& 1 - \frac{|D|}{|C^cD|}\frac{1 - e^{-\lambda |C^cD|}}{1-e^{-\lambda|D|}}e^{-\lambda|C|}\,;
\end{align*}
so we have
{\small
\begin{equation*}
1 - \frac{|D|}{|C^cD|}e^{-\lambda|C|} \leq \p{\phi'(C) > 0 ~\Big|~ \phi(D)>0, X} \leq 1 - e^{-\lambda|C|}\,.
\end{equation*}}
Using the observation above and \eqref{eq:dist2} we obtain \eqref{eq:inter_b}.
%%%%%%%%%%%%%%%%%%%%%%%%%%%%%%%%%%%%%%%%%%%%%%%%%%%%%%%%%%%%%%%%%%%%%%%%%%%%%%%%%%%%%%%%%%%%%%%%%%%%
% Where \eqref{eq:inter_b} is located.
%%%%%%%%%%%%%%%%%%%%%%%%%%%%%%%%%%%%%%%%%%%%%%%%%%%%%%%%%%%%%%%%%%%%%%%%%%%%%%%%%%%%%%%%%%%%%%%%%%%%
Therefore,
\begin{equation}
\ee{\frac{\phi'(CD)}{\phi'(C)}~\big|~\phi'(C)>0, \phi(D)>0, X} < \frac{|CD|}{|C|}\,. \nonumber
\end{equation}

Noting that
\begin{equation*}
1 - \frac{1}{a} \leq ae^{-a}\int_0^1 ye^{ay}\D{y} = 1 - \frac{1 - e^{-a}}{a} \leq 1\,,
\end{equation*}
we could derive \eqref{eq:ineq2} from \eqref{eq:inter_a} for large enough $N$ such that $1 - \frac{1}{\lambda |D|} - \frac{|C||D|}{|CD||C^cD|}\exp(-2\lambda |CD^c|) > 0$.
%%%%%%%%%%%%%%%%%%%%%%%%%%%%%%%%%%%%%%%%%%%%%%%%%%%%%%%%%%%%%%%%%%%%%%%%%%%%%%%%%%%%%%%%%%%%%%%%%%%%
% Where \eqref{eq:ineq2} is located.
%%%%%%%%%%%%%%%%%%%%%%%%%%%%%%%%%%%%%%%%%%%%%%%%%%%%%%%%%%%%%%%%%%%%%%%%%%%%%%%%%%%%%%%%%%%%%%%%%%%%
Hence we can ascertain that
\begin{align*}
&\ee{\frac{\phi'(CD)}{\phi'(C)}~\big|~\phi'(C)>0, \phi(D)>0, X} > \\
&\left(1 - \frac{1}{\lambda |D|} - \frac{|C||D|}{|CD||C^cD|}e^{-2\lambda |CD^c|}\right)\frac{|CD|}{|C|}\,.
\end{align*}

As such, the selected point is less likely to be in $D$ than the case where we assume $\Phi'$ is Poisson on $C$.% (or that the points in $C$ truly were uniformly distributed).

\section{Derivation of inequality \eqref{eq:expect}}
\label{appendix:expecting}

We have $(x'_n,y'_n) \eqd (Rz\cos(\theta),Rz\sin(\theta))$, where $\theta\sim\mbox{Uniform}(-\pi/2,\pi/2)$ and $z\sim\mbox{Beta}(2,1)$ are independent. Thus, we have
\begin{subequations}
\label{eq:increment moments}
\begin{align}
\ee{x'_n}&=R\frac{2}\pi\int_0^{\pi/2}\cos(\theta)\,\D{\theta}\int_0^12z^2\,\D{z}
=\frac{4R}{3\pi}\,, \label{eq:first moment}\\
\ee{(y'_n)^2} &= \frac{R^2}\pi\int_{-\pi/2}^{\pi/2}\sin^2(\theta)\,\D{\theta}\int_0^12z^3\,\D{z}
=\frac{R^2}{4}\,. \label{eq:second moment}
\end{align}
\end{subequations}

Also, by first changing $x$ to $1-x$ and then using polar coordinates, we obtain
{\allowdisplaybreaks
\begin{align*}
\frac1R\,\textrm{E}\bigg(g(R,&x'_n,y'_n)\bigg)+1 \\
&=\frac{4}{\pi}\int_0^1\int_0^11_{x^2+y^2\le1}\sqrt{(1-x)^2+y^2}\,\D{x}\D{y} \cr
%&=\frac{4}{\pi}\int_0^1\int_0^11_{(1-x)^2+y^2\le1}\sqrt{x^2+y^2}\,\D{x}\D{y} \cr
%&=\frac{2}{\pi}\int_0^{\pi/4}\int_0^{\sec\theta}2z^2\,\D{z}\D{\theta} ~+ \\
%&~\quad \frac2\pi\int_{\pi/4}^{\pi/2}\int_0^{2\cos\theta}2z^2\,\D{z}\D{\theta} \cr
&=\frac{4}{3\pi}\int_0^{\pi/4}
\Bigl((\sec(\theta))^3+(2\sin(\theta))^3\Bigr)\,\D{\theta} \cr
&=\frac{3(2^{3/2})+6\log(1+\sqrt{2})+64-5(2^{7/2})}{9\pi} \\
&\approx 0.7499728\,.
\end{align*}}

Hence, $\ee{g(R,x'_n,y'_n)}<-\frac{R}4\,$.

\ifCLASSOPTIONcaptionsoff
  \newpage
\fi
\end{document}